\documentclass[11pt]{article}
\usepackage{algorithm}
\usepackage{algpseudocode}
\algnewcommand\algorithmicinput{\textbf{Input:}}
\algnewcommand\algorithmicoutput{\textbf{Output:}}
\algnewcommand\Input{\item[\algorithmicinput]}
\algnewcommand\Output{\item[\algorithmicoutput]}

\usepackage{amssymb,amsmath,amsthm, mathtools}
\usepackage{graphicx}
\usepackage{tikz}
\usetikzlibrary{calc,shapes.multipart,chains,arrows,positioning}
\newtheorem{theorem}{Theorem}
\newtheorem{corollary}{Corollary}
\newtheorem{lemma}{Lemma}
\newtheorem{proposition}{Proposition}
\usepackage{a4wide}
\date{}
\begin{document}
\title{Linear Time LexDFS on Cocomparability Graphs}
\author{Ekkehard K\"{o}hler%
\thanks{Diskrete Mathematik \& Grundl. der Informatik, Brandenburg University of Technology, 03044 Cottbus, Germany (ekoehler@math.tu-cottbus.de)} \and Lalla Mouatadid\thanks{Computer Science Department,University of Toronto, Toronto, ON M55 3G4, Canada (lalla@cs.toronto.edu)}}
\maketitle
\begin{abstract}
Lexicographic depth first search (LexDFS) is a graph search protocol which has already proved to be a powerful tool on cocomparability graphs. Cocomparability graphs have been well studied by investigating their complements (comparability graphs) and their corresponding posets. Recently however LexDFS has led to a number of elegant polynomial and near linear time algorithms on cocomparability graphs when used as a preprocessing step \cite{1,9,11}. The nonlinear runtime of some of these results is a consequence of complexity of this preprocessing step. We present the first linear time algorithm to compute a LexDFS cocomparability ordering, therefore answering a problem raised in \cite{1} and helping achieve the first linear time algorithms for the minimum path cover problem, and thus the Hamilton path problem, the maximum independent set problem and the minimum clique cover for this graph family. 
\end{abstract}
\section{Introduction}
Graph searching is a very useful and widely used tool that gave rise to a number of efficient and easily implementable algorithms. Lexicographic breadth first search (LexBFS) for instance, is a well known graph search protocol which has led to elegant algorithms on various graph families, as illustrated in \cite{2,12}. Recently another graph search protocol, lexicographic depth first search (LexDFS), was introduced and has already proved to be a powerful tool on cocomparability graphs \cite{1,9,11}. Indeed, since LexDFS was introduced \cite{3}, many problems, such as computing a maximum cardinality independent set, a minimum clique cover or a minimum path cover, now have near-linear time solutions for cocomparability graphs. These successful approaches share a strategy: They start with creating a so-called cocomparability ordering of the graph and preprocess it with a LexDFS sweep and then basically extend, or slightly modify, the linear time algorithms that work for interval graphs (a subfamily of cocomparability graphs). The nonlinear runtime of the algorithms on cocomparability graphs is is forced by the nonlinearity of LexDFS.

In this paper, we present the first linear time algorithm to compute a LexDFS cocomparability ordering. Therefore, as immediate corollaries, we now have the first linear time algorithm to compute a minimum path cover, and thus a Hamilton path if one exists, and the first linear time algorithm to compute a maximum independent set and a minimum clique cover for cocomparability graphs. We will also show how to specifically compute a LexDFS$^+$ ordering in linear time. LexDFS$^+$ is a variant of LexDFS that is needed in \cite{1,9,11}.

Cocomparability graphs are a family of perfect graphs whose complements, comparability graphs, admit a transitive orientation of the edges. That is, for every three vertices $x, y, z$, if the edges $xy, yz$ are oriented $x \rightarrow y \rightarrow z$ then $xz \in E$ and $x \rightarrow z$. Cocomparability graphs and partially ordered sets, or $\emph{posets}$, are closely related.  In Section 2, we explain this relationship and how algorithms on cocomparability graphs immediately translate into algorithms on posets.  Cocomparability graphs have been well studied by investigating their complements (comparability graphs) and their corresponding posets \cite{4}. Recently however, there has been a growing motivation to exploit the structure of cocomparability graphs in order to design algorithms that do not require the computation of the complement or the poset. For this approach the use of LexDFS has proven to be quite sucessful \cite{1,9,11}. 

A key point in our algorithm for computing a LexDFS ordering is a partition refinement approach of the layers of a corresponding poset of the cocomparability graph, $\emph{without}$ computing the poset itself. This refinement is $\emph{in situ}$ and performed $\emph{backwards}$, as opposed to the well known forward partition refinement used to compute a LexBFS ordering \cite{12}. We discuss the technique of partition refinement in more details in Section 2. 

The paper is organized as follows: Section 2 gives the necessary background and relevant definitions. In Section 3, we present the LexDFS algorithm, and in 4 we prove its correctness and show how to compute in linear time a LexDFS$^+$ ordering. In Section 5 we present our concluding remarks. Due to a lack of space, we leave the implementation details to the Appendix.
\section{Background}
We assume the reader to be familiar with basic graph notation. All the graphs considered in this paper are finite, simple, and undirected, unless explicitly stated otherwise. For a vertex $v$, $N(v) = \{u | uv \in E\}$; and we say $v$ is $\emph{simplicial}$ if $N(v)$ is a clique. An $ordering$ $\sigma$ of $V$ is a bijection $\sigma$: $[1...n] \rightarrow V$. Given an ordering $\sigma = v_1, v_2, ..., v_n$, we write $v_i \prec_{\sigma} v_j$ if $v_i$ appears before $v_j$ in $\sigma$, i.e., $i < j$; and $N^+(v_i) = \{v_j | v_iv_j \in E$ and $i < j\}$; we denote by $\sigma^- = v_n, v_{n-1}, ..., v_2, v_1$ the reverse ordering of $\sigma$. 
 
Given a cocomparability graph $G(V, E)$, an ordering $\sigma$ is a $\emph{cocomparability}$ $\emph{ordering}$ (or an umbrella free ordering) of $G$ if for any triple $a \prec_{\sigma} b \prec_{\sigma} c$ where $ac \in E$, we either have $ab \in E$ or $bc \in E$, or both \cite{5}. If neither $ab$ and $bc$ are edges, we say that the umbrella $ac$ $\emph{flies}$ over $b$. Note that the cocomparability ordering is just the equivalent to transitivity in the complement. In \cite{6}, McConnell and Spinrad presented an algorithm that computes a cocomparability ordering in $\mathcal{O}(m+n)$ time, where $m = |E|$ and $n = |V|$.
 
A poset $P(V, \prec)$ is an irreflexive, antisymmetric and transitive relation on the set V. We say that two elements $a, b\in V$ are comparable if $a \prec b$ or $b \prec a$, otherwise they are incomparable. A linear extension $\mathcal{L}$ of $P$ is a total ordering of $V$ which respects the order imposed by $\prec$. As already mentioned, posets and cocomparability graphs are closely related. In fact, if $G(V, E)$ is a comparability graph, then $G$ together with a transitive orientation of $E$ can equivalently be represented by a poset $P(V, \prec)$ where $uv \in E$ if and only if $u$ and $v$ are comparable in $P$. This implies that $\sigma$, a cocomparability ordering of $\overline{G}$ is a linear extension of $P$. Notice that for every poset $P(V, \prec)$, there exists a unique comparability graph $G(V, E)$, and thus a unique cocomparability graph. Conversely, every transitive orientation of a comparability graph, and thus every cocomparability ordering in the complement, is a linear extension of $\emph{a}$ poset $P$. 
 
A graph search is a mechanism for visiting vertices of a given graph in a certain manner. We say that two or more vertices are $\emph{tied}$ if at a given step of the graph search, these vertices are all eligible to be visited next. In 2008, Corneil and Krueger \cite{3} introduced LexDFS, a graph search that extends depth first search by assigning lexicographic labels to the vertices in order to break ties. Algorithm 1 is the generic LexDFS algorithm, as presented in \cite{3} and Fig. 1 is a step by step illustrative example. LexDFS admits the following vertex ordering characterization, known as the $\emph{4 Point Condition}$: 
\begin{theorem}\cite{3}
$\sigma$ is a LexDFS ordering of a graph $G(V, E)$ if and only if for every triple $a \prec_{\sigma} b \prec_{\sigma} c$ where $ac \in E, ab \notin E$, there must exists a vertex $d$ such that $a \prec_{\sigma} d \prec_{\sigma} b$ and $db \in E, dc \notin E$. 
\end{theorem}
\begin{algorithm}[H]
\caption{LexDFS}
\begin{algorithmic}[1]
\Input a graph $G(V, E)$ and a start vertex $s$
\Output a LexDFS ordering $\sigma$ of $V$
	\State assign the label $\epsilon$ to all vertices
	\State $label(s) \leftarrow \{0\}$
	\For{$i \leftarrow 1$ to \emph{n}}
		\State pick an unnumbered vertex $v$ with lexicographically largest label
		\State $\sigma(i) \leftarrow v$ \Comment{$v$ is assigned the number $i$}
		\ForAll{unnumbered vertex $w$ adjacent to $v$}
			\State prepend $i$ to $label(w)$
		\EndFor
	\EndFor
\end{algorithmic}
\end{algorithm}
\begin{figure}[H]
\begin{minipage}{.3\textwidth}
\centering
\begin{tikzpicture}
  \node[circle, draw, fill=white!50] (n1) at (0,1) {$a$};
  \node[circle, draw, fill=white!50] (n2) at (1.5,1) {$d$};
  \node[circle, draw, fill=white!50] (n3) at (1.5,0) {$c$};
  \node[circle, draw, fill=white!50] (n4) at (0,0) {$b$};
  \node[circle, draw, fill=white!50] (n6) at (1.5,-1) {$e$};
  \foreach \from/\to in {n1/n2, n1/n3, n1/n4, n2/n3, n3/n6, n3/n4}
    \draw (\from) -- (\to);
\end{tikzpicture}
\end{minipage}
\begin{minipage}{0pt}
\centering
\begin{tabular}{lllll}
\hline\noalign{\smallskip}
    		{\small $\sigma(i)$} & {\small Affected Vertices} & {$\sigma$}\\ 
    		\hline\noalign{\smallskip}
    		\parbox[t]{1in}{$\sigma(1) = a$}&\parbox[t]{2in}{$label(b) = label(c) = label(d) =1$\strut} & a\\
    		\parbox[t]{1in}{$\sigma(2) = b$}&\parbox[t]{2in}{$label(c) = 21$} & a, b\\
    		\parbox[t]{1in}{$\sigma(3) = c$}&\parbox[t]{2in}{$label(d) = 31$ \par $label(e) = 3$\strut} & a, b, c\\
    		\parbox[t]{1in}{$\sigma(4) = d$}& & a, b, c, d\\
    		\parbox[t]{1in}{$\sigma(5) = e$}& & a, b, c, d, e\\
\hline
\end{tabular}
\end{minipage}
\caption{$G(V, E)$ a cocomparability graph, and a step by step computation of a LexDFS ordering of $G$ starting at vertex $a$.}
\end{figure}

LexDFS$^+$ is the LexDFS variant with the additional $`\emph{rightmost}$' tie breaking rule. That is, given a vertex ordering $\sigma$ of $G$, the ordering $\tau = $ LexDFS$^+(\sigma)$ is a LexDFS of $G$ where ties between eligible vertices are broken by choosing the rightmost vertex in $\sigma$. Therefore, by definition, LexDFS$^+$ always starts by the rightmost vertex in $\sigma$. For an example, look at the graph in Fig. 1. If we compute $\tau = $LexDFS$^+(\sigma)$, for this example, we have to start with vertex $e$, i.e. $\tau(1) = e$; obviously, $\tau(2) = c$. For $\tau(3)$ in a regular LexDFS $a, b$ and $d$ are tied. However, for LexDFS$^+(\sigma)$, $\tau(3) = d$, since $d$ is rightmost in $\sigma$ among $a, b,$ and $d$. It was shown in \cite{1} that if $G(V, E)$ is a cocomparability graph, and $\sigma$ a cocomparability ordering of $G$, then the LexDFS ordering $\tau = LexDFS^+(\sigma)$ is also a cocomparability ordering of $G$. It is easy to see that if $\sigma$ is a LexDFS cocomparability ordering, then for any triple $a \prec_{\sigma} b \prec_{\sigma} c$ where $c$ is a nonsimplicial vertex and $ab \notin E, ac, bc \in E$, there exists a vertex $d$ such that $a \prec_{\sigma} d \prec_{\sigma} b$ and ad, $db \in E$ and $dc \notin E$ \cite{1}. Indeed the edge $ad$ destroys the umbrella $ac$ over $d$. 

As was already pointed out in the introduction, the key idea of our algorithm to determine a LexDFS of a cocomparability graph is the backward in situ partition refinement of the layers of a corresponding poset, without computing the poset itself. Given a set $S$, we call $\mathcal{P} = (P_1, P_2, ..., P_k)$ a partition of $S$ if for all $P_i, P_j$, $i \neq j$, $P_i \cap P_j = \emptyset$ and $\bigcup_{i=1}^k P_i = S$. Given a set $T \subseteq S$, we say that $T$ refines $\mathcal{P}$ when every partition class $P_i \in \mathcal{P}$ is replaced with subpartition classes $A_i = P_i  \cap  T$ and $B_i = P_i \backslash A_i$. This technique, known as $\emph{partition refinement}$ has led to a simple and elegant implementation of LexBFS in linear time \cite{12}. The LexBFS partition refinement algorithm is as follows: Initially $\mathcal{P} = (V)$; select a start vertex $s$ where $N(s)$ refines $\mathcal{P}$ by placing $A = V \cap N(s)$ before $B = V\backslash A$. The vertex whose neighbourhood is used to refine the partition classes is called a $\emph{pivot}$. Pick the next pivot $v$ amongst vertices in $A$; and use $N(v)$ for refining $A$ then $B$ and maintaining the order of the partition classes created so far: $(A \cap N(v)$, $A\backslash N(v)$, $B\cap N(v)$, $B\backslash N(v))$ in this order. This process is repeated until all partition classes have been refined.

This refinement can be seen as a $\emph{forward}$ refinement in the sense that pivots are selected left to right, i.e., from the $A$'s sets then the $B$'s set, and the refinement is $\emph{in situ}$, meaning the $A_i$'s always precede the $B_i$'s. In other words, pivots do not reorder the already created subpartitions of $P$. That is , if $N(u)$ was used to refine $P$ to $A = P \cap N(u)$ and $B = P \backslash A$, and $v$ is the next pivot then $N(v)$ is used to refine $A$ first to $A \cap N(v)$ followed by $A \backslash N(v)$, next $N(v)$ refines $B$ to $B \cap N(v)$ followed by $B \backslash N(v)$, and the subpartitions of $A$ always precede the subpartitions of $B$. This in situ refinement results in a linear time implementation for LexBFS. Also for LexDFS, one can define a partition refinement scheme, but this partition refinement is not in situ. Consider a pivot $v$; due to the depth first search character of LexDFS, $v$ has to pull $\emph{all}$ its neighbours to the front, i.e., $A\cap N(v)$ followed by $B\cap N(v)$ both precede $A\backslash N(v)$ followed by $B\backslash N(v)$. This sorting of the partition classes is an obstacle to a linear time implementation for LexDFS. In fact, to the best of our knowledge, there is no linear time implementation of LexDFS. The best known algorithm takes $\mathcal{O}(\min(n^2, n + m\log \log n))$ time and uses the above explained non in situ partition refinement together with van Emde Boas trees \cite{7}.
\section{The Algorithm}
Before presenting our algorithm in detail we first give an overview. Let $G(V, E)$ be a cocomparability graph. We first compute a cocomparability ordering $\sigma$ using the algorithm in \cite{6}. Then, we assign a label, denoted $\#(v)$, to each vertex $v$ in $\sigma$. We use these labels to compute, for every vertex $v$, the length of a largest chain succeeding $v$ in the corresponding poset of the complement of $G$. Roughly speaking, we then partition $V$ by iteratively placing vertices with smallest label into the same partition set. In a comparability graph one could finds these sets by iteratively removing the set of maximal elements of the poset.  Since we work on the complement, this has to be done using only edges of the cocomparability graph, i.e. non-edges of the comparability graph. Once all the vertices have their initial labeling, we iteratively create a partition $\mathcal{P}$ of $V$ wherein each step $i$, the partition class $P_i$ consists of the vertices of minimum label value. When a vertex $v$ is added to a partition class $P_i$, we say that $v$ has been $\emph{visited}$.

Since $\sigma$ is a cocomparability ordering, and thus a linear extension of a poset, the $P_1$ vertices are exactly the elements in the linear extension with no upper cover. Therefore they are just the maximal elements of the partial order defined by $\sigma$ in $\overline{G}$. Similarly, when all the $P_1$ have been visited, i.e., `removed', $P_2$ is just the set of maximal elements in the partial order of $\overline{G} \backslash P_1$, and so on. Creating the partition classes is indeed equivalent to removing the maximal elements of a poset corresponding to $\overline{G}$ one layer at a time. 

The final step of the algorithm is the partition refinement where we refine each partition class one at a time in a specific manner. In particular, each partition class $P_i$ is assigned a set $S_i$ of pivots that will be used to refine $P_i$ only. The set $S_i$ is implemented as a stack, and the order in which the pivots are pushed onto $S_i$ is crucial. When $v$ is taken from $S_i$ to be the next pivot, $N(v)$ performs an in situ refinement on $P_i$. We use $\tau_i$ to denote the final (refined) ordering of $P_i$. When all partition classes have been refined, we concatenate all the $\tau_i$'s in order, i.e., $\tau = \tau_1 \cdot \tau_2 \cdot ... \cdot \tau_p$ where $\cdot$ denotes concatenation, and use $\tau$ to denote the final ordering. Our main theorem is the following: 
\begin{theorem}
Let $G(V, E)$ be a cocomparability graph, $\tau$ is a LexDFS cocomparability ordering that can be computed in $\mathcal{O}(m+n)$. 
\end{theorem}
We next discuss each step of the algorithm in more detail and prove the correctness of the algorithm in Section 4. Due to space constraints, we leave the implementation details to the Appendix.
\subsection{Vertex Labelling}
Let $G(V, E)$ be an undirected cocomparability graph, and let $\sigma = v_1 \prec_{\sigma} v_2 \prec_{\sigma} ... \prec_{\sigma} v_n$ be a cocomparability ordering of $G$ returned by the algorithm in \cite{6}; we use $\sigma = ccorder(G)$ to denote such an algorithm. For every vertex $v \in V$, we assign a label $\#(v)$ initialized to the number of nonneighbours of $v$ to its right in $\sigma$: $\#(v) = |\{u | uv \notin E \mbox{ and } v \prec_{\sigma} u\}|$. 

Given such a labelling of the vertices, we create a partition of V denoted by $\mathcal{P} = \bigcup_{i=1}^p P_i$ in the following manner: Initially all vertices are marked unvisited, $P_1$ is the set of vertices with the smallest $\#$ label value. Now all vertices in $P_1$ are marked to be visited. For all unvisited vertices $u$ and for all $v \in P_1$, such that $uv \in E$, $\#(u)$ is incremented by 1. To create $P_2$, again select the set of unvisited vertices of smallest $\#$ value. These vertices in $P_2$ are marked to be visited and for each such $v \in P_2$ and unvisited $u$ adjacent to $v$, $\#(u)$ is incremented by 1. We increment $i$ and repeat this operation of creating a partition class of the vertices with the smallest label until all vertices belong to a partition class. 
\begin{algorithm}[H]
\caption{PartitionClasses}
\begin{algorithmic}[1]
\Input a cocomparability graph $G(V, E)$
\Output partition $\mathcal{P}$ of V with $p$ partition classes, and $p$
		\State $\sigma \leftarrow ccorder(G(V, E))$\Comment{As computed in \cite{6}}
		\State $S \leftarrow \emptyset$
		\For{$i \leftarrow \emph{n}$ downto $1$}
			\State $\#(v_i) \leftarrow (n-i) - |S\cap N(v_i)|$\Comment{Initial labelling $\#(v)$}
			\State $S \leftarrow S \cup \{v_i\}$
		\EndFor
		\State $U \leftarrow V$\Comment{$U$ the set of unvisited vertices}
		\State $i \leftarrow 1$
		\While{$U$ not empty}
			\State $P_i \leftarrow \{ v | \#(v) = min(\#(U))\}$\Comment{Creating Partition Classes}
			\State $U \leftarrow U\backslash P_i$
			\For{$v \in P_i$}
				\For{$u \in U$ and $uv \in E$}
					\State $\#(u) \leftarrow \#(u)+1$
				\EndFor
			\EndFor
			\State $i \leftarrow i+1$
		\EndWhile
		\State $p \leftarrow i-1$
		\State \textbf{return} $\mathcal{P} \leftarrow (P_1, P_2, ..., P_p)$ and $p$
\end{algorithmic}
\end{algorithm}
Algorithm 2 is a formal description of the algorithm which takes a cocomparability graph $G(V, E)$ as input and returns the partition $\mathcal{P} = (P_1, P_2, ..., P_p)$. Let $\pi = P_1 \cdot P_2 \cdot ... \cdot P_p$ be the order of V resulting from Algorithm 2 such that $\forall x \in P_i, y \in P_{j>i}$,  we have $x \prec_{\pi} y$. The order inside each $P_i$ is arbitrary. Consider the graph in Fig. 2 with a valid cocomparability vertex ordering. The numbers below the vertices are their labels as computed by Algorithm 2. Table 1 illustrates the creation of the partition classes.
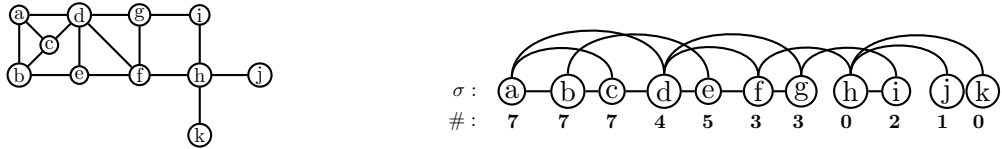
\begin{figure}[H]
\begin{minipage}{.5\textwidth}
\centering
\begin{tikzpicture}[thick,scale=0.8, every node/.style={scale=0.7}]
  \node[circle, draw, fill=white!50, inner sep=1pt, minimum width=4pt] (n1) at (0,1) {a};
  \node[circle, draw, fill=white!50, inner sep=1pt, minimum width=4pt] (n2) at (0,0) {b};
  \node[circle, draw, fill=white!50, inner sep=1pt, minimum width=4pt](n11) at (0.5, 0.5) {c};
  \node[circle, draw, fill=white!50, inner sep=1pt, minimum width=4pt] (n3) at (1,1) {d};
  \node[circle, draw, fill=white!50, inner sep=1pt, minimum width=4pt] (n4) at (1,0) {e};
  \node[circle, draw, fill=white!50, inner sep=1pt, minimum width=4pt] (n5) at (2,0) {f};
  \node[circle, draw, fill=white!50, inner sep=1pt, minimum width=4pt] (n6) at (2,1) {g};
  \node[circle, draw, fill=white!50, inner sep=1pt, minimum width=4pt] (n8) at (3,1) {i};
  \node[circle, draw, fill=white!50, inner sep=1pt, minimum width=4pt] (n7) at (3,0) {h};
  \node[circle, draw, fill=white!50, inner sep=1pt, minimum width=4pt] (n9) at (4,0) {j};
  \node[circle, draw, fill=white!50, inner sep=1pt, minimum width=4pt] (n0) at (3,-1) {k};
  \foreach \from/\to in {n1/n2, n1/n3, n2/n4, n3/n4, n3/n5, n3/n6, n4/n5, n5/n6, n5/n7, n6/n8, n7/n8, n7/n9, n7/n0, n11/n1, n11/n2, n11/n3}
    \draw (\from) -- (\to);
\end{tikzpicture}
\end{minipage}
\begin{minipage}{.2\textwidth}
\centering
\begin{tikzpicture}[thick,scale=0.8, every node/.style={scale=0.7}]
  	\node[anchor=east] at (-0.7,0.5) (01111) {$\#:$};
  	\node[anchor=east] at (0,0.5) (1) {\bf{7}};
  	\node[anchor=west] at (0.4,0.5) (2) {\bf{7}};
  	\node[anchor=west] at (1.2,0.5) (3) {\bf{7}};
  	\node[anchor=west] at (2,0.5) (4) {\bf{4}};
  	\node[anchor=west] at (2.8,0.5) (5) {\bf{5}};
  	\node[anchor=west] at (3.6,0.5) (6) {\bf{3}};
  	\node[anchor=west] at (4.3,0.5) (7) {\bf{3}};
  	\node[anchor=west] at (5.1,0.5) (8) {\bf{0}};
  	\node[anchor=west] at (5.9,0.5) (9) {\bf{2}};
  	\node[anchor=west] at (6.7,0.5) (0) {\bf{1}};
  	\node[anchor=west] at (7.3,0.5) (11) {\bf{0}};
  	\node[anchor=east] at (-0.7,1) (00) {$\sigma:$};
	\tikzstyle{every node}=[circle, draw, fill=white!50, inner sep=1pt, minimum width=4pt]
	\node[anchor=east] at (0,1) (1) {a};
  	\node[anchor=west] at (0.4,1) (2) {b};
  	\node[anchor=west] at (1.2,1) (3) {c};
  	\node[anchor=west] at (2,1) (4) {d};
  	\node[anchor=west] at (2.8,1) (5) {e};
  	\node[anchor=west] at (3.6,1) (6) {f};
  	\node[anchor=west] at (4.3,1) (7) {g};
  	\node[anchor=west] at (5.1,1) (8) {h};
  	\node[anchor=west] at (5.9,1) (9) {i};
  	\node[anchor=west] at (6.7,1) (0) {j};
  	\node[anchor=west] at (7.3,1) (11) {k};
  	\draw (1) edge[out=90,in=90,-] (3);
  	\draw (1) edge[out=90,in=90,-] (4);
  	\draw (2) edge[out=90,in=90,-] (5);
  	\draw (4) edge[out=90,in=90,-] (6);
  	\draw (4) edge[out=90,in=90,-] (7);
  	\draw (6) edge[out=90,in=90,-] (8);
  	\draw (7) edge[out=90,in=90,-] (9);
  	\draw (8) edge[out=90,in=90,-] (0);
  	\draw (8) edge[out=90,in=90,-] (11);
	\foreach \from/\to in {1/2, 2/3, 3/4, 4/5, 5/6, 6/7, 8/9}
    \draw (\from) -- (\to);
\end{tikzpicture}
\end{minipage}
\caption{$G(V, E)$, a cocomparability order $\sigma$ of V, and the initial labelling $\#$ of V.}
\end{figure}
\vspace{-0.9cm}
\begin{table}[H]
\centering
\caption{Creation of the partition classes for the graph in Fig. 2.}
\begin{tabular}{llll}
\hline\noalign{\smallskip}
$P_i$ &smallest label & $v \in P_i$ & Labels to increment \\
\noalign{\smallskip}
\hline
\noalign{\smallskip}
		$P_1$  & 0 & \{h, k\} &\parbox[t]{1.5in}{$\#(j)=\#(j)+1=2$ \par $\#(i)=\#(i)+1=3$ \par $\#(f)=\#(f)+1=4$\strut}  \\ 
    		\parbox[t]{0.5in}{$P_2$}  & \parbox[t]{1in}{2} & \parbox[t]{1in}{\{j\}} &  \\ 
		$P_3$ & 3 & \{g, i\} & \parbox[t]{1.5in}{$\#(d)=\#(d)+1=5$ \par $\#(f)=\#(f)+1=5$\strut} \\ 
		$P_4$ & 5 & \{d, e, f\} & \parbox[t]{1.5in}{$\#(a)=\#(a)+1=8$ \par $\#(b)=\#(b)+1=8$ \par $\#(c)=\#(c)+1=8$\strut}  \\ 
		$P_5$ & 8 & \{a, b, c\} & \\
\hline
\end{tabular}
\end{table}
In the remainder of the paper, we will use $\#^*(v_i)$ to refer to the initial value of $v_i$'s label, i.e., the number of nonneighbours of $v$ to its right in $\sigma$; and $\#_k(v_i)$ to denote the label value of $v_i$ when $P_k$ is being created, i.e., at iteration $k$. 
\subsection{Partition Refinement}
Once all the partition classes are computed, we reorder the adjacency list of each $v$ according to $\pi$ in $\mathcal{O}(m + n)$ time (see Appendix), then construct a new ordering of V by refining $\mathcal{P}$. Our refinement (algorithm $\emph{Refine}$) is slightly different than the generic partition refinement algorithm presented in \cite{8} and briefly explained in Section 2.
\begin{algorithm}[h]
\caption{Refine}
\begin{algorithmic}[1]
\Input a partition class $P$ ordered by $\pi$ and its corresponding ordered list of pivots $S$
\Output refinement $\tau$ of $P$
	\State $Q_1 \leftarrow P$, $k \leftarrow 1$
	\While{$S$ not empty}
		\State $j \leftarrow 1$
		\State $v \leftarrow S.pop$\Comment{$S$ is implemented as a stack}
		\For{$i \leftarrow 1 $ to $k$}
			\If{$|Q_i \cap N(v)| = 0$ or $|Q_i \cap N(v)| = |Q_i|$}
				\State $Q'_j \leftarrow Q_i$
				\State $j \leftarrow j+1$
			\Else
				\State $Q'_j \leftarrow Q_i \cap N(v)$
				\State $Q'_{j+1} \leftarrow Q_i\backslash N(v)$
				\State $j \leftarrow j+2$
			\EndIf
		\EndFor
		\State $k \leftarrow j-1$
		\For{$i \leftarrow 1$ to $k$}\Comment{Rename the new partitions for the next pivot}
			\State $Q_i \leftarrow Q'_i$
		\EndFor
	\EndWhile
	\State \Return $\tau \leftarrow Q_1\cdot Q_2\cdot ... \cdot Q_k$\Comment{$x \in Q_i, y \in Q_{j>i} \implies x \prec_{\tau} y \mbox{ and }  x, y \in Q_i \implies x \prec_{\tau} y \mbox{ iff } x \prec_{\pi} y$}
\end{algorithmic}
\end{algorithm}
Given the partition $\mathcal{P} = (P_1, P_2, ..., P_p)$ returned by Algorithm 2, we associate a set $S_i$ to each $P_i$, where $S_i$ is a set of pivots that will be used to refine $P_i$. We say that a partition class $P_i$ is $\emph{processed}$ when it has been refined. We use $\tau_i$ to denote the final ordering of $P_i$ after it has been refined, i.e. $\tau_i = Refine(P_i, S_i)$. If a partition class $P_i$ has an empty pivot set $S_i$, then for $\tau_i$ the (arbitrary) ordering of $P_i$ in $\pi$ is used.

The sets $S_i$ are implemented as stacks and are created as follows: $S_1 = \emptyset$ and $P_1$ is considered processed. For all $P_{i>1}$, we scan $\tau_1$ from $\emph{left to right}$ and for each $v \in \tau_1$ and every $u \in P_{i>1}$ where $uv \in E$, we push $v$ in $S_i$. In general, every time a partition class $P_{j<i}$ is refined, i.e., $\tau_j$ has been produced, we scan $\tau_j$ from left to right, and for every $v \in \tau_j$ with neighbours in $P_{i>j}$, we push $v$ into $S_i$. To refine $P_i$, we pop elements of $S_i$ one at a time, and for each $v \in S_i$, $v$ is the pivot that refines $P_i$ by reordering $P_i$ into the subpartitions $P_i \cap N(v)$ followed by $P_i \backslash N(v)$. The next pivot out of $S_i$ performs an $\emph{in situ}$ refinement of the current subpartitions of $P_i$. 
\begin{algorithm}
\caption{UpdatePivots}
\begin{algorithmic}[1]
\Input a newly refined partition class $P_j$ and its index $j$
\Output updated pivot lists for the upcoming partition classes, i.e. for $P_i$, $i>j$
	\For{$v \in P_j$}\Comment{in the $\tau_j$ order}
		\If{$v$ has neighbours in $P_{i > j}$}\State $S_{i}.push(v)$\Comment{Update the pivot list of $P_i$}\EndIf
	\EndFor
\end{algorithmic}
\end{algorithm}

Let $u_j^1, u_j^2, ..., u_j^k$ be the left to right ordering of the vertices inside $\tau_j$. We mentioned in Section 1 that not only this refinement is in situ, but also $\emph{backwards}$. Backwards in two ways: First, the pivots of $\tau_{j}$ have a higher priority, i.e. a stronger pull, than the pivots of $\tau_{k<j}$, and second the pivots are pushed down the stack $S_i$ in $\tau_1 \cdot \tau_2 \cdot ... \cdot \tau_{i-1}$ order (left to right) and thus are popped in reverse order. Therefore we maintain the priority of the pivots in the backward order:  $(\tau_1 \cdot \tau_2 \cdot ... \cdot \tau_{i-1})^-$, but also the priority of the pivots inside each $\tau_{j<i}$ in the backward order $\tau_{j}$. That is for any two vertices $u_j^a, u_j^b \in \tau_j$ where $a < b$, if $u_j^a$ and $u_j^b$ are both pivots for $P_i$, then $u_j^b$ refines $P_i$ first before $u_j^{a<b}$. Note that this is very similar to standard partition refinement with the difference that in standard partition refinement, $P_i$ is first refined by $\tau_1$ then $\tau_2, \tau_3$ and so on. Here we start refining with the last vertex in $\tau_{i-1}$, then $\tau_{i-2}$, and so on up to $\tau_1$. This opposite refinement shows the key difference between LexDFS and LexBFS. Whereas in a LexBFS order the earliest neighbours have the strongest pull and the latest neighbours the weakest, in a LexDFS the last vertices are more influencial then the earlier visited ones. Algorithm 3, $\emph{Refine}$, takes $P_i$ and $S_i$ as input, and returns the new ordering $\tau_i$ of $P_i$. Algorithm 4, $\emph{UpdatePivots}$, takes $\tau_j$ as input, the refined ordering of $P_j$, and updates the stacks $S_i$ for all unprocessed partition classes $P_{i>j}$.

Let $\tau$ denote the final ordering of all the refined partition classes, i.e., $\tau = \tau_1 \cdot \tau_2 \cdot ... \cdot \tau_p$. Using the graph again in Fig. 2, we show in Table 2 the in situ refinement of each $P_i$. The final ordering is $\tau = h, k, j, i, g, f, d, e, b, c, a$.
\begin{table}[h]
\centering
\caption{The refinement of the $P_i$'s constructed in Table 1.}
\begin{tabular}{lllll}
\hline\noalign{\smallskip}
    		{\small $P_i$ ordered by $\pi$} & {\small $S_i$} & {\small Pivots} & {\small Refinement} & {\small$\tau_i$}\\ 
    		\hline\noalign{\smallskip}
    		\parbox[t]{1in}{$P_1$: \{h, k\}} &\parbox[t]{0.5in}{$\emptyset$}&\parbox[t]{0.5in}{-} &\parbox[t]{1in}{(h, k)} & h, k \\
    		\parbox[t]{1in}{$P_2$: \{j\}} &\parbox[t]{0.5in}{h}&\parbox[t]{0.5in}{h} &\parbox[t]{1in}{(j)} & j \\    		
    		\parbox[t]{1in}{$P_3$: \{g, i\}} &\parbox[t]{0.5in}{h}&\parbox[t]{0.5in}{h} &\parbox[t]{1in}{((i)(g))} & i, g \\
		\parbox[t]{1in}{$P_4$: \{d, e, f\}} &\parbox[t]{0.5in}{g, h}&\parbox[t]{0.5in}{$g$ \par $h$\strut} &\parbox[t]{1in}{((d, f)(e))\par((f)(d)(e))}\strut & f, d, e \\
		\parbox[t]{1in}{$P_5$: \{c, a, b\}} &\parbox[t]{0.5in}{e, d}&\parbox[t]{0.5in}{$e$ \par $d$\strut} &\parbox[t]{1in}{((b)(c, a))\par((b)(c, a))}\strut & b, c, a\\
		\hline
    	\end{tabular}
\end{table}
\subsection{The Complete Algorithm}
We are now ready to present the complete algorithm $\emph{CCLexDFS}$. 
\begin{algorithm}
\caption{CCLexDFS}
\begin{algorithmic}[1]
\Input a cocomparability graph $G(V, E)$ 
\Output a LexDFS order $\tau$ of $G$ that is also a cocomparability order of $G$
	\State $\tau \leftarrow \emptyset$
	\State $(\mathcal{P}, p) \leftarrow PartitionClasses(G)$\Comment{Compute the partition classes}
	\State $S_1, ..., S_p \leftarrow \emptyset$
	\For{$i \leftarrow 1$ to $p$}
		\State $\tau_i  \leftarrow Refine(P_i, S_i)$\Comment{Refine the partition classes}
		\State$UpdatePivots(\tau_i, i)$\Comment{Update the pivot sets}
		\State $\tau \leftarrow \tau \cdot\tau_i$
	\EndFor
	\State \textbf{return} $\tau$
\end{algorithmic}
\end{algorithm}
\section{Correctness of the Algorithm}
We denote the set of partition classes $P_1$ to $P_{i-1}$ by $\mathcal{P}_{i} = (P_1, P_2, ..., P_{i-1})$.
\begin{lemma}
	For any $u, v$ such that $u \prec_{\sigma} v$ and $uv \notin E$: ${\#^*}(u) > {\#^*}(v)$ and at any step $i$, $\#_i(u) > \#_i(v)$.
\end{lemma}
\begin{proof}
	Since $u \prec_{\sigma} v$ and $uv \notin E$, any vertex $w$ with $v \prec_{\sigma} w, wv \notin E $ implies $wu \notin E$ otherwise $uw$ flies over $v$ contradicting $\sigma$ being a 	cocomparability ordering. Thus $w$ contributes equally to ${\#^*}(u)$ and ${\#^*}(v)$. Moreover, since $u \prec_{\sigma} v, uv \notin E$, $v$ also contributes to ${\#^*}(u)$. Therefore 			${\#^*}(u) > {\#^*}(v)$.

	Suppose at a step $i, \#_i(u) < \#_i(v)$ then a vertex $z \in P_{j<i} \in \mathcal{P}_i$ must have closed the gap in ${\#^*}(u) > {\#^*}(v)$ by contributing to $\#_j(v)$ but not to $\#_j(u)$. 			Let $z \in P_{j<i}$ be such a vertex, we are only interested in the case when $zv \in E$ and $zu \notin E$, since adding $z$ to $P_j \in \mathcal{P}_{i}$ would have incremented $\#_j(v)$, 			making it closer to $\#_j(u)$. Notice that $u \prec_{\sigma} z$, otherwise $zv$ flies over $u$ which contradicts $\sigma$ being a cocomparability ordering . Therefore $z$ contributes 			one to $u$'s label as well, namely ${\#^*}(u)$, thus not reducing the gap between $u$ and $v$'s labels. Therefore $\#_i(u) > \#_i(v)$.
\end{proof}
\begin{lemma}
	For $1 \le i \le p$, $P_i$ is the set of maximal elements in the poset $P \backslash \bigcup_{j=1}^{i-1}P_j$.
\end{lemma}
\begin{proof}
	The proof is by induction on the partition classes $P_i$. $P_1$ is the set of vertices that are adjacent to all the vertices to their right in $\sigma$. Since $\sigma$ is a cocomparability, it is a linear extension of a poset $P(V, \prec)$; thus $P_1$ is the set of elements in $P$ that are incomparable to all the elements to their right in the linear extension $\sigma$ of $P$. Hence the elements of $P_1$ are maximal elements in the poset $P$. 
	
	Assume for contradiction that $P_i$ is the first partition class where the vertices of $P_i$ are not the maximal elements of the poset $P \backslash \bigcup_{j=1}^{i-1}P_j$; in particular let $v$ denote an element in $P_i$ that is not maximal in $P \backslash \bigcup_{j=1}^{i-1}P_j$. This means there exists an element $u$ where $v \prec u$ and $u$ is an element of $P \backslash \bigcup_{j=1}^{i-1}P_j$. Since $\sigma$ is a linear extension of $P$ and $v \prec u$, $v \prec_{\sigma} u$; and since $u$ and $v$ are comparable, they are non adjacent in $G$. Therefore, $v \prec_{\sigma} u$ and $uv \notin E$. By Lemma 1, $\#_i(v) > \#_i(u)$ and thus $v \notin P_i$.

\end{proof}
\begin{lemma}
Every partition class $P_i \in \mathcal{P}$ returned by Algorithm 1 is a clique.
\end{lemma}
\begin{proof}
By Lemma 2, every $P_i$ is the set of the maximal elements in the poset $P \backslash \bigcup_{j=1}^{i-1} P_j$, and maximal elements in the poset form an antichain, i.e., an independent set in the comparability graph. Thus it is a clique in the cocomparability graph. 
\end{proof}
\begin{lemma} (The Flipping Lemma): Let $\sigma$ be a cocomparability order, and $\tau$ the corresponding ordering created from $\sigma$ and returned by Algorithm 5. For every $uv \notin E$, $u \prec_{\sigma} v \iff v \prec_{\tau} u$.
\end{lemma}
\begin{proof}
As we are assigning vertices to their partition classes, let $u$ and $v$ be the left most pair of vertices in $\tau$ to satisfy $uv \notin E$ and $u \prec_{\sigma} v$ and $u \prec_{\tau} v$. By Lemma 1, ${\#^*}(u) > {\#^*}(v)$, and by Lemma 3, $u$ and $v$ belong to two different partition classes; $P_i$ and $P_{j>i}$ respectively since $u \prec_{\tau} v$. Therefore when $P_i$ was created $\#_i(u) < \#_i(v)$, which contradicts Lemma 1. Therefore $v \prec_{\tau} u$.

For sufficiency, using the contraposition we know that $(u \prec_{\tau} v \Rightarrow v \prec_{\sigma} u)$ if and only if $(u \prec_{\sigma} v \Rightarrow v \prec_{\tau} u)$. Thereby completing the proof.
\end{proof}
\begin{corollary}
$\tau$ is a cocomparability order of $G$.
\end{corollary}
\begin{proof}
As in \cite{1}, if $\tau$ is not a cocomparability order as witnessed by $x \prec_{\tau} y \prec_{\tau} z$ and $xz \in E, xy, yz \notin E$, then by the Flipping Lemma we have $y \prec_{\sigma} x$ and $z \prec_{\sigma} y$, which implies that the umbrella $zx$ flies over $y$ in $\sigma$, contradicting the fact that $\sigma$ is a cocomparability order.
\end{proof}
We are now ready to prove Theorem 2. Namely, that the ordering $\tau$ produced by Algorithm 5 is a LexDFS cocomparability order of $G$. For implementation details, we refer the reader to the Appendix.
\begin{proof}[Proof of Theorem 2]
By Corollary 1, we know that $\tau$ is a cocomparability order of $G$. Suppose it is not a LexDFS order. Therefore for some triple $a, b, c$ with $a \prec_{\tau} b \prec_{\tau} c$, $ac \in E$ and $ab \notin E$, there doesn't exist a vertex $d$ as required by the 4 Point Condition (Theorem 1). Since $\tau$ is a cocomparability order, $bc \in E$ to destroy the umbrella $ac$ over $b$.

Suppose $b$ and $c$ belong to the same partition class $P_i$. Since $ab\notin E$ and $a \prec_{\tau}b$, $a \in P_{j<i}$. Since $ac \in E$, $a$ is a pivot with respect to $P_i$ and thus $b$ and $c$ could not have been in the same subclass since $a$ would have pulled $c$ before $b$. Since $b \prec_{\tau} c$, a pivot $u$ that pulled $b$ in front of $c$ must exist, i.e. $ub \in E, uc \notin E$. Pick $u$ to be the rightmost pivot to $b$ to satisfy this configuration. Notice that $a \prec_{\tau} u$. Otherwise, since the refinement is backwards, $a$ would have refined $P_i$ before $u$ thus pulling $c$ in front of $b$. Therefore $a \prec_{\tau} u$. This means that $u$ plays the role of $d$ with respect to LexDFS, a contradiction to our assumption; therefore $b$ and $c$ must be in different partition classes.

Since $b \prec_{\tau} c$, $b \in P_i$ and $c \in P_{j>i}$, when $P_i$ was created $\#_i(b) < \#_i(c)$. We investigate how this gap could have occurred given $ac, bc \in E$ and $ab \notin E$. Without loss of generality, let $a, b, c$ be the left most triple in $\tau$ that does not satisfy the 4 Point Condition, and consider the vertices that have contributed to $\#_i(b)$ and $\#_i(c)$. Let $u$ be one of these vertices. Thus $u$ increased $\#_i(b), \#_i(c)$ either by being a non adjacent right neighbour of $b$ or $c$ in the initial ordering $\sigma$ (Algorithm 2, line 6) or $u$ changed $\#(b), \#(c)$ when $u$ was assigned to a partition class (Algorithm 2, lines 12-16). 

If $u$ is contained in a set of $\mathcal{P}_{i}$, then $u \prec_{\tau} b \prec_{\tau} c$. Consider all the possible adjacencies between $u, b$ and $c$. If $ub \in E$ and $uc \notin E$, then either $a \prec_{\tau} u$ in which case $u$ plays the role of $d$ as required by LexDFS; or $u \prec_{\tau} a$ in which case $u$ increments $\#_j(b)$ at iteration $j$ when $u$ was assigned to $P_{j < i}$, a set in $\mathcal{P}_{i}$, and $u$ also contributes to ${\#^*}(c)$ by the Flipping Lemma. Thus $u$ contributes equally to the labels of $b$ and $c$. If $ub, uc \in E$ then $u$ increments both $b$'s and $c$'s labels when it is assigned to a partition class; and if $ub, uc \notin E$, then by the Flipping Lemma, $u$ contributes to both $\#^*(b)$ and $\#^*(c)$. But in all three cases $u$ does not reduce the gap between $b$ and $c$'s labels. Therefore $ub \notin E, uc \in E$. However by the Flipping Lemma $b \prec_{\sigma} u$, and thus $u$ contributes to $\#^*(b)$ since $ub \notin E$, but also $u$ increments $c$'s label since $u \in \mathcal{P}_{i}$ and $uc \in E$; again not reducing the gap. Therefore $u$ must be in $V \backslash \mathcal{P}_{i}$.
	
If $u$ is not in a set of $\mathcal{P}_{i}$, then $u$ has not been assigned to a partition class yet. Since $u$ is responsible for the gap $\#_i(b) < \#_i(c)$, $u$ created this gap when $b$ and $c$ were assigned their initial labels $\#^*(b)$ and $\#^*(c)$. For $u$ to create such a gap, $u$ must contribute to $c$'s initial label $\#^*(c)$ and not contribute to $\#^*(b)$. In other words, $uc \notin E$ and $c \prec_{\sigma} u$. Therefore by the Flipping Lemma, $u \prec_{\tau} c$. Moreover, for $u$ to not contribute to $\#^*(b)$, we either have $ub \in E$ or $ub \notin E$ but $u \prec_{\sigma} b$. Notice that this latter case is impossible, since $u \prec_{\sigma} b \Rightarrow b \prec_{\tau} u$ (by the Flipping Lemma); but also $u \prec_{\tau} c$ causing $bc$ to fly over $u$ in $\tau$ and contradicting Corollary 1. Thus $ub \in E, uc \notin E$ and $c \prec_{\sigma} u$. Since $u$ is not in a set of $\mathcal{P}_{i}$, $b \prec_{\tau} u$, otherwise $u$ plays the role of $d$ with respect to LexDFS. Moreover, $au \in E$ since $\tau$ is a cocomparability order; and the triple $a, b, u$ must satisfy the LexDFS ordering otherwise we contradict the choice of $a, b, c$ as $u \prec_{\tau} c$. Therefore there must exist a vertex $w$ such that $a \prec_{\tau} w \prec_{\tau} b, wb \in E, wu \notin E$; this forces the edge $aw$ in order to avoid the umbrella $au$ over $w$. If $wc \notin E$, then $w$ plays the role of $d$ as required by LexDFS for the triple $a, b, c$, and if $wc \in E$, then the umbrella $wc$ flies over $u$, contradicting $\tau$ being a cocomparability order. Therefore there must always exists a vertex that satisfies the LexDFS ordering for $\#_i(b) < \#_i(c)$ to hold; and thus $\tau$ is a LexDFS cocomparability order of $G$. 
\end{proof}
\begin{corollary}
Prior to the partition refinement step, if the vertices inside each partition class were ordered according to $\sigma^{-}$, then the resulting $\tau$ is a $LexDFS^{+}(\sigma)$.
\end{corollary}
\begin{proof}
Suppose vertices inside each $P_i$ are ordered according to $\sigma^{-}$, clearly the resulting ordering $\tau$ is a LexDFS, since Theorem 2 holds for any ordering inside the partition classes. Suppose $\tau$ is not a LexDFS$^+$.

That is, at some iteration $i$, there exists two tied vertices $u$ and $v$ where $u \prec_{\sigma} v$ but $u \prec_{\tau} v$. $u, v$ are tied if $N(u) \cap \mathcal{P}_{i} = N(v) \cap \mathcal{P}_{i}$ and $u, v \in P_i$, i.e., they are ready to be processed. Since $u \prec_{\sigma} v$ and vertices inside $P_i$ are ordered according to $\sigma^{-}$, $v$ is left of $u$ inside $P_i$. Since $N(u) \cap \mathcal{P}_{i} = N(v) \cap \mathcal{P}_{i}$, there doesn't exist a pivot $w \in \mathcal{P}_{i}$ to pull $u$ in front of $v$ when processing $P_i$. Therefore $v \prec_{\tau} u$. $\tau$ is a LexDFS$^{+}$.
\end{proof}
\section{Conclusion and Open Problems}
We have presented the first linear time algorithm to determine a LexDFS cocomparability order, therefore answering a question raised in \cite{1}, and also overcoming the bottleneck in the near linear time algorithms in \cite{1,9}. It is still an open question whether there exists a linear time implementation for LexDFS on arbitrary graphs. Our implementation exploits the poset structure of the cocomparability graph. In fact, computing the partition classes is equivalent to computing the layers of the corresponding poset. It is fairly straightforward to see that if $G(V, E)$ is a $\emph{comparability}$ graph, a LexDFS($\overline{G}$) cocomparability ordering can also be computed in time linear in the size of $G$. The details to extend the algorithm to compute such an ordering will be provided in the journal paper. Clearly this leads to the obvious question of whether this algorithm can also be modified to compute a LexDFS ordering of a comparability graph.

Looking at the power of LexDFS on cocomparability, and how it has led to simple and elegant algorithms on this graph family when LexBFS has failed, simply by extending the existing algorithms on interval graphs; it is natural to ask whether there are other problems that can be solved using a similar approach: First a LexDFS preprocessing, then extending the algorithm for interval graphs.

Moreover with this algorithm in hand now, preprocessing is `easy', which raises the question of possible multisweep LexDFS algorithms. Multisweeps algorithms perform a constant number of sweeps (i.e., graph searches) where each sweep generally reveals more structural properties about the graph. LexDFS has not been used in a multisweep manner yet, we therefore raise the question of whether a second LexDFS sweep reveals more structure that was not seen through the previous sweep. If so, are there problems that can benefit from this structure? 

Stepping away from cocomparability graphs but still looking at structured graph families, it is natural to ask whether LexDFS can be implemented in linear time for other restricted graph families, such as asteroidal triple free graphs, a graph family that contains cocomparability graphs. But also whether there are other applications to LexDFS in other graph classes. Graph searches have been exploited on various graph families, it is therefore necessary to explore the possible insights LexDFS has to offer, in contrast with these other graph searches.
\subsubsection*{Acknowledgments.} The authors thank Derek Corneil for his helpful suggestions and valuable comments.
\bibliographystyle{splncs}

%
%
\newpage
\section*{Appendix}
Recall that $\#^*(v_i)$ refers to the initial value of $v_i$'s label, i.e., the number of nonneighbours of $v$ to its right in $\sigma$; and $\#(v_i)$ is the label value of $v_i$ if it was incremented at some previous iteration.
\subsubsection{The Partition Classes:}
We assume we are given an adjacency list representation of $G$, where each adjacency list is a doubly linked list. The ordering $\sigma$ is also implemented as a doubly linked list. Every vertex $v_i$ is represented with a node structure containing a variable $pos$ to store $v_i$'s position $i$ in $\sigma$, a variable indicating $v_i$'s label value $(\#(v_i))$, a pointer $pc$ initialized to NULL to indicate $v_i$'s partition class, and a pointer $pt$ to $v_i$'s adjacency list. 
\begin{lemma}
Let $G(V, E)$ be a cocomparability graph. The algorithm \\$PartitionClasses$ takes $\mathcal{O}(m + n)$ time to compute the partition classes of $G$.
\end{lemma}
\begin{proof}
Step 1 of $PartitionClasses$ is computed in $\mathcal{O}(m + n)$ time \cite{6}. Consider step 3: Computing the vertices' labels can be done in linear time by scanning the adjacency list of each $v_i$ and keeping track of the number of vertices $v_j$ where $j > i$. Then the initial label value is just $\#^*(v_i) = n - i - |\bigcup_{v_{i}v_{j>i} \in E} v_j|$. For each vertex $v_i$, we scan its adjacency list in $\mathcal{O}(d_{v_i})$ steps where $d_{v_i}$ is the degree of $v_i$. Therefore steps 3 to 6 take $\mathcal{O}(m + n)$ time. 

For the while loop in line 9 to 18, consider a set of bins $B_0, ..., B_{n-1}$ together with corresponding counter variables $c_0, ..., c_{n-1}$, where $c_i$'s value gives the number of elements in $B_i$. Initially all bins are empty and all $c_i = 0$. Each bin uses a doubly linked list to store its vertices. When vertex $v$'s initial label is computed, $v$ is placed in bin $B_i$ where $i = \#^*(v)$, and $c_i$ is incremented by one. Once all vertices are assigned to a bin, we use these bins to create the $p$ partition classes $P_1, ..., P_p$ in linear time as follows: We process the nonempty bins one at a time in increasing order of their indices. For every bin $B_i$ where $B_i$ is the left most unprocessed bin with at least one vertex (i.e., $c_i \neq 0$), we place the vertices of $B_i$ in $P_j$ where $j \ge 1$ is the next unused index. For every vertex $v$ placed in $P_j$ we update $v$'s $pc$ variable to point to $P_j$, and examine each vertex $u$ adjacent to $v$. If $u \in B_k$, $k > i$, then $u$ is moved to $B_{k+1}$ and $c_k$ is decremented by one and $c_{k+1}$ is incremented by one; since the bins are doubly linked lists, the insertion and deletion of vertices take constant time. Therefore  when processing a bin, for each one of its vertices $v$, we do at most $\mathcal{O}(d_v)$ increments and decrements for the neighbours of $v$, and thus creating the partition classes takes at most $\mathcal{O}(m + n)$ time over all vertices. 
\end{proof}

Once all the partition classes are created, we reorder the adjacency list of each $v$ according to $\pi$ in $\mathcal{O}(m + n)$ time. To do so, we create new adjacency lists as follows: We go through the vertices according to $\pi$'s ordering and for every vertex $v \in \pi$ and every $u$ adjacent to $v$, we add $v$ to $u$'s new adjacency list, and we update $v$'s pointer $pos$ to point to $v$ in $\pi$. Fig. 3 below illustrates this construction for vertex $g$ of the graph in Fig. 2 with neighbours restricted to partition class $P_4$. 
\newpage
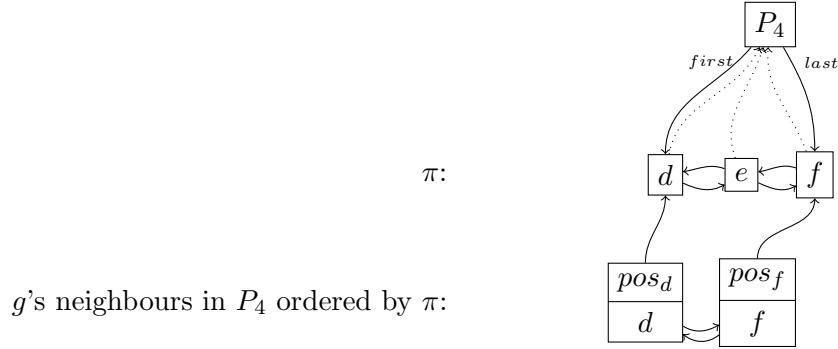
\begin{figure}[H]
\centering
	\begin{tikzpicture}
	\node[ rectangle split, rectangle split parts=1, draw,anchor=east] at (1.5,4) (k) {\nodepart{one} $P_4$};
	\node[anchor=east] at (0.8,3.5) (first) {{\tiny $first$}};
	\node[anchor=east] at (2.2,3.5) (l) {{\tiny $last$}};
	
	\node[anchor=east] at (-3,2) (00) {$\pi$:};
  	\node[ rectangle split, rectangle split parts=1, draw, anchor=east] at (0,2) (d) {\nodepart{one} $d$};
  	\node[ rectangle split, rectangle split parts=1, draw, anchor=east] at (1,2) (f) {\nodepart{one} $e$};
  	\node[ rectangle split, rectangle split parts=1, draw, anchor=east] at (2,2) (g) {\nodepart{one} $f$};

	\node[anchor=east] at (-3,0.3) (00) {$g$'s neighbours in $P_4$ ordered by $\pi$:};
  	\node[ rectangle split, rectangle split parts=2, draw,   anchor=east] at (0,0.3) (a) {\nodepart{second} $d$ \nodepart{one} $pos_d$};
  	\node[ rectangle split, rectangle split parts=2, draw,   anchor=east] at  (1.5,0.3) (b) {\nodepart{second} $f$ \nodepart{one} $pos_f$};
 
	\draw (k) edge[out=-130, in=90,->](d); 
	\draw (k) edge[out=-60, in=90,->](g);

	\draw[dotted] (d) edge[out=80, in=-115, ->](k);
	\draw[dotted] (f) edge[out=110, in=-107, ->](k);
	\draw[dotted] (g) edge[out=110, in=-95, ->](k);

	\draw (d) edge[out=-25, in=210, ->](f);
	\draw (f) edge[out=160, in=10,->](d);

	\draw (f) edge[out=-25, in=210, ->](g);
	\draw (g) edge[out=160, in=10,->](f);

	\draw (a) edge[out=90, in=-90,->](d);
	\draw (b) edge[out=90, in=-90,->](g);

	\draw (a) edge[out=-32,in=210,->] (b);
	\draw(b) edge[out=220, in=-40,->](a);
	\end{tikzpicture}
  \caption{The new adjacency list of vertex $g$ restricted to the partition class $P_4$ as constructed in Table 1.}
\end{figure}
\subsubsection{The Partition Refinement:}
We implement $\emph{Refine}$ using a similar data structure as the one presented in \cite{8}, where $\pi$ is a doubly linked list and each $P$ is implemented as a doubly linked list with three pointers: one to its first ($\bf{f}$) element in $\pi$, one to its last ($\bf{l}$), and one denoted $\bf{c}$ (for current), initialized to NULL. Every vertex $v$ in $\pi$ points back to the partition class it belongs to. It is important to remember when the new adjacency lists were created, vertices in the adjacency lists point back to their image in $\pi$, as it was just illustrated in Fig. 3.

When analyzing $\emph{Refine}$, we are under the assumption that the set of pivots $S$ is given such that every $v \in S$ must have at least one neighbour in $P$. This is accomplished by Algorithm 4, $\emph{UpdatePivots}$ as follows: Every $v$ knows the partition class to which it belongs, and the adjacency lists are sorted according to $\pi$, thus every time a partition class $P_j$ is refined, we sweep through $\tau_j$ from left to right, and for each $v \in \tau_j$, with neighbours in $P_{i>j}$, $v$ is pushed onto $S_i$. This operation takes at most $\mathcal{O}(d_v)$ steps per $v \in P_j$. Summing over all $v \in P_j$, it is easy to see that $UpdatePivots$ is linear in the number of edges between $P_j$ and any class $P_{i>j}$. 

Now consider Algorithm 3, steps 5 to 15 in $\emph{Refine}$: Every time we pop an element $v$ out of $S$, we know that there must exist at least one (sub)partition of $P$ (denoted $Q_i$ in $\emph{Refine}$) which contains a neighbour of $v$. When $v$ is the pivot, we scan its adjacency list (but only the neighbours in $P$), in the order imposed by $\pi$. For each $w$ adjacent to $v$ where $w$ has its image in a subpartition $Q_i$, we check if the pointer $\bf{c}$ of $Q_i$ is NULL. If $\bf{c}$ is NULL, $w$ is moved to the start of $Q_i$ where the $\bf{f}$ pointer is updated accordingly, and $\bf{c}$ points to $w$ as well. If $\bf{c}$ is not NULL, $w$ is placed right after the vertex that $\bf{c}$ points to, then $\bf{c}$ is updated to point to $w$. Notice if $w$ is the last element in $Q_i$ and has been moved, then the $\bf{l}$ pointer of $Q_i$ is also updated accordingly. When all the neighbours of the pivot have been processed, we split the reordered (sub)partitions at every location where $\bf{c}$ is not NULL to reflect the new reordered subpartitions. Each newly created subpartitions has its $\bf{c}$ reset to NULL, and its first and last pointers updated. It's important to note that if a subpartition $Q_i$ has $\bf{f} = \bf{l} = \bf{c}$, then we know that it has a singleton element and thus cannot be refined more, therefore its $\bf{c}$ pointer is not reset to NULL. 

For further clarity on how the partition refinement is performed, we consider the following example. Suppose $\emph{Refine}$ takes a partition class $P$ and a set of pivots $S$ such that:
\begin{align*}
	P = \{b, a, f, e, d\}\\
	S = \{x, y, z\}\\
	N(x) \cap P = \{b, a, e, d\}\\
	N(y) \cap P = \{a, f, e\}\\
	N(z) \cap P = \{b, a\}
\end{align*}

Suppose the ordering in each set is fixed (as would be the case when $\pi$ imposes its ordering), and pivots are popped out of $S$ in the order of $x$ then $y$ then $z$. Fig. 4, 5 and 6 illustrate how each pivot refines $P$ (i.e., by reordering the vertices, and splitting at the pointer $\bf{c}$). 
\begin{figure}[H]
\begin{minipage}{.3\textwidth}
	\begin{tikzpicture}[thick,scale=0.6, every node/.style={scale=0.7}]
	\node[ rectangle split, rectangle split parts=1, draw,   anchor=east] at (0.5,4) (k) {\nodepart{one} $P$};
	\node[anchor=east] at (0.5,5) (l) {{\tiny null}};
	\node[anchor=east] at (1.08,4.5) (ll) {{\tiny $\bf{c}$}};

	\node[anchor=east] at (-0.5,3.5) (first) {{\tiny $\bf{f}$}};
	\node[anchor=east] at (3.35,3) (last) {{\tiny $\bf{l}$}};

	\node[ rectangle split, rectangle split parts=1, draw,   anchor=east] at (-1,2) (e) {\nodepart{one} $b$};
  	\node[ rectangle split, rectangle split parts=1, draw,   anchor=east] at (0,2) (d) {\nodepart{one} $a$};
  	\node[ rectangle split, rectangle split parts=1, draw,   anchor=east] at (1,2) (f) {\nodepart{one} $f$};
  	\node[ rectangle split, rectangle split parts=1, draw,   anchor=east] at (2,2) (g) {\nodepart{one} $e$};
	\node[rectangle split, rectangle split parts=1, draw,   anchor=east] at  (3.3,2) (h) {\nodepart{one} $d$};

	\node[ circle, draw, anchor=east] at (-2.2,0.3) (00) {$x$};
  	\node[ rectangle split, rectangle split parts=1, draw,   anchor=east] at (-1.5,0.3) (x) {\nodepart{one} $b$};
  	\node[ rectangle split, rectangle split parts=1, draw,   anchor=east] at (0,0.3) (a) {\nodepart{one} $a$};
  	\node[ rectangle split, rectangle split parts=1, draw,   anchor=east] at  (1.5,0.3) (b) {\nodepart{one} $e$};
   	\node[ rectangle split, rectangle split parts=1, draw,   anchor=east] at  (3,0.3) (y) {\nodepart{one} $d$};

	\draw (k) edge[out=-130, in=90,->](e); 
	\draw (k) edge[out=-60, in=90,->](h); 
	\draw (k) edge[out=90, in=-90, ->](l);

	\draw (00) edge[out=0, in=180, ->](x);

	\draw[dotted] (e) edge[out=110, in=-115, ->](k);
	\draw[dotted] (d) edge[out=80, in=-115, ->](k);
	\draw[dotted] (f) edge[out=110, in=-107, ->](k);
	\draw[dotted] (g) edge[out=110, in=-95, ->](k);
	\draw[dotted] (h) edge[out=110, in=-95, ->](k);

	\draw (e) edge[out=-25, in=210, ->](d);
	\draw (d) edge[out=160, in=10,->](e);

	\draw (d) edge[out=-25, in=210, ->](f);
	\draw (f) edge[out=160, in=10,->](d);

	\draw (f) edge[out=-25, in=210, ->](g);
	\draw (g) edge[out=160, in=10,->](f);

	\draw (g) edge[out=-25, in=210, ->](h);
	\draw (h) edge[out=160, in=10,->](g);

	\draw (a) edge[out=90, in=-90,->](d);
	\draw (b) edge[out=90, in=-90,->](g);
	\draw (x) edge[out=90, in=-90,->](e);
	\draw (y) edge[out=90, in=-90,->](h);

	\draw (a) edge[out=-32,in=210,->] (b);
	\draw(a) edge[out=220, in=-40,->](x);
	\draw(x) edge[out=-32, in=210, ->](a);
	\draw(b) edge[out=220, in=-40,->](a);
	\draw (b) edge[out=-32,in=210,->] (y);
	\draw(y) edge[out=220, in=-40,->](b);
	\node[anchor=east] at (3,6) (capt) {The initial setup for pivot $x$};
\end{tikzpicture}
\end{minipage}
\begin{minipage}{.4\textwidth}
\centering
	\begin{tikzpicture}[thick,scale=0.6, every node/.style={scale=0.7}]
	\node[anchor=east] at (0.5,4.5) (kkk) {};
	\node[ rectangle split, rectangle split parts=1, draw,   anchor=east] at (0.5,4) (k) {\nodepart{one} $P$};
	\node[anchor=east] at (-0.5,3.5) (first) {{\tiny $\bf{f}$}};
	\node[anchor=east] at (2.4,2.5) (l) {{\tiny $\bf{c}$}};
	\node[anchor=east] at (3.1,3) (last) {{\tiny $\bf{l}$}};

	\node[ rectangle split, rectangle split parts=1, draw,   anchor=east] at (-1,2) (e) {\nodepart{one} $b$};
  	\node[ rectangle split, rectangle split parts=1, draw,   anchor=east] at (0,2) (d) {\nodepart{one} $a$};
  	\node[ rectangle split, rectangle split parts=1, draw,   anchor=east] at (1,2) (f) {\nodepart{one} $e$};
  	\node[ rectangle split, rectangle split parts=1, draw,   anchor=east] at (2,2) (g) {\nodepart{one} $d$};
	\node[rectangle split, rectangle split parts=1, draw,   anchor=east] at  (3.3,2) (h) {\nodepart{one} $f$};
	\node[ circle, draw, anchor=east] at (-2.2,0.3) (00) {$x$};
  	\node[ rectangle split, rectangle split parts=1, draw,   anchor=east] at (-1.5,0.3) (x) {\nodepart{one} $b$};
  	\node[ rectangle split, rectangle split parts=1, draw,   anchor=east] at (0,0.3) (a) {\nodepart{one} $a$};
  	\node[ rectangle split, rectangle split parts=1, draw,   anchor=east] at  (1.5,0.3) (b) {\nodepart{one} $e$};
   	\node[ rectangle split, rectangle split parts=1, draw,   anchor=east] at  (3,0.3) (y) {\nodepart{one} $d$};
	\draw (00) edge[out=0, in=180, ->](x);
	\draw (k) edge[out=-130, in=90,->](e); 
	\draw (k) edge[out=-60, in=90,->](h); 
	\draw (k) edge[out=-90, in=90, ->] (g);

	\draw[dotted] (e) edge[out=110, in=-115, ->](k);
	\draw[dotted] (d) edge[out=80, in=-115, ->](k);
	\draw[dotted] (f) edge[out=110, in=-107, ->](k);
	\draw[dotted] (g) edge[out=110, in=-95, ->](k);
	\draw[dotted] (h) edge[out=110, in=-95, ->](k);

	\draw (e) edge[out=-25, in=210, ->](d);
	\draw (d) edge[out=160, in=10,->](e);

	\draw (d) edge[out=-25, in=210, ->](f);
	\draw (f) edge[out=160, in=10,->](d);

	\draw (f) edge[out=-25, in=210, ->](g);
	\draw (g) edge[out=160, in=10,->](f);

	\draw (g) edge[out=-25, in=210, ->](h);
	\draw (h) edge[out=160, in=10,->](g);

	\draw (a) edge[out=90, in=-90,->](d);
	\draw (b) edge[out=90, in=-90,->](f);
	\draw (x) edge[out=90, in=-90,->](e);
	\draw (y) edge[out=90, in=-90,->](g);

	\draw (a) edge[out=-32,in=210,->] (b);
	\draw(a) edge[out=220, in=-40,->](x);
	\draw(x) edge[out=-32, in=210, ->](a);
	\draw(b) edge[out=220, in=-40,->](a);
	\draw (b) edge[out=-32,in=210,->] (y);
	\draw(y) edge[out=220, in=-40,->](b);
	\node[anchor=east] at (3, 6) (capt) {The new reordering of $P$};
\end{tikzpicture}
\end{minipage}
\begin{minipage}{.1\textwidth}
	\begin{tikzpicture}[thick,scale=0.6, every node/.style={scale=0.7}]
	\node[anchor=east] at (0.5,0) (kkk) {};
	\node[ rectangle split, rectangle split parts=1, draw,   anchor=east] at (0.5,4) (k) {\nodepart{one} $Q_1$};
	\node[anchor=east] at (0.5,5) (l) {{\tiny null}};
	\node[anchor=east] at (-0.23,3.5) (first) {{\tiny $\bf{f}$}};
	\node[anchor=east] at (1,3.2) (last) {{\tiny $\bf{l}$}};
	\node[anchor=east] at (0.5,4.5) (current) {{\tiny $\bf{c}$}};

	\node[ rectangle split, rectangle split parts=1, draw,   anchor=east] at (3,4) (kk) {\nodepart{one} $Q_2$};
	\node[anchor=east] at (3,5) (ll) {{\tiny null}};
	\node[anchor=east] at (2.65,3.5) (firstb) {{\tiny $\bf{f}$}};
	\node[anchor=east] at (3.25,3.5) (lastb) {{\tiny $\bf{l}$}};
	\node[anchor=east] at (3,4.5) (currentb) {{\tiny $\bf{c}$}};

	\node[ rectangle split, rectangle split parts=1, draw,   anchor=east] at (-1,2) (e) {\nodepart{one} $b$};
  	\node[ rectangle split, rectangle split parts=1, draw,   anchor=east] at (0,2) (d) {\nodepart{one} $a$};
  	\node[ rectangle split, rectangle split parts=1, draw,   anchor=east] at (1,2) (f) {\nodepart{one} $e$};
  	\node[ rectangle split, rectangle split parts=1, draw,   anchor=east] at (2,2) (g) {\nodepart{one} $d$};
	\node[rectangle split, rectangle split parts=1, draw,   anchor=east] at  (3.3,2) (h) {\nodepart{one} $f$};

	\draw (k) edge[out=-130, in=90,->](e); 
	\draw (k) edge[out=-90, in=90, ->] (g);
	\draw (k) edge[out=90, in=-90, ->](l);

	\draw (kk) edge[out=260, in=180, ->] (h);
	\draw (kk) edge[out=-75, in=0, ->] (h);
	\draw (kk) edge[out=90, in=-90, ->](ll);

	\draw[dotted] (e) edge[out=110, in=-115, ->](k);
	\draw[dotted] (d) edge[out=80, in=-115, ->](k);
	\draw[dotted] (f) edge[out=110, in=-107, ->](k);
	\draw[dotted] (g) edge[out=110, in=-95, ->](k);

	\draw[dotted] (h) edge[out=110, in=-95, ->](kk);

	\draw (e) edge[out=-25, in=210, ->](d);
	\draw (d) edge[out=160, in=10,->](e);

	\draw (d) edge[out=-25, in=210, ->](f);
	\draw (f) edge[out=160, in=10,->](d);

	\draw (f) edge[out=-25, in=210, ->](g);
	\draw (g) edge[out=160, in=10,->](f);
	\node[anchor=east] at (3, 6) (capt) {Splitting $P$ at $\bf{c}$};
\end{tikzpicture}
\end{minipage}
\caption{Processing pivot $x$}
\end{figure}
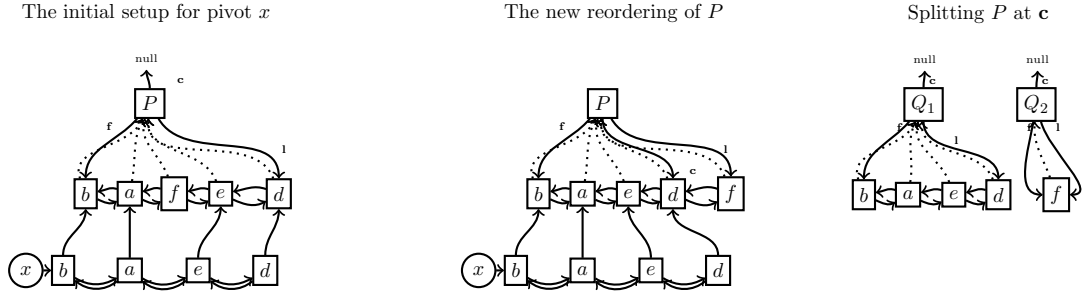
The next pivot $y$ will refine according to the newly created subpartitions as follows:
\begin{figure}[H]
\begin{minipage}{.3\textwidth}
	\begin{tikzpicture}[thick,scale=0.6, every node/.style={scale=0.7}]
	\node[anchor=east] at (0.5,6) (kkk) {};
	\node[ rectangle split, rectangle split parts=1, draw,   anchor=east] at (0.5,4) (k) {\nodepart{one} $Q_1$};
	\node[anchor=east] at (0.5,5) (l) {{\tiny null}};
	\node[anchor=east] at (-0.2,3.5) (first) {{\tiny $\bf{f}$}};
	\node[anchor=east] at (1,3.2) (last) {{\tiny $\bf{l}$}};
	\node[anchor=east] at (0.5,4.5) (current) {{\tiny $\bf{c}$}};

	\node[ rectangle split, rectangle split parts=1, draw,   anchor=east] at (3,4) (kk) {\nodepart{one} $Q_2$};
	\node[anchor=east] at (3,5) (ll) {{\tiny null}};
	\node[anchor=east] at (2.7,3.5) (firstb) {{\tiny $\bf{f}$}};
	\node[anchor=east] at (3.25,3.5) (lastb) {{\tiny $\bf{l}$}};
	\node[anchor=east] at (3,4.5) (currentb) {{\tiny $\bf{c}$}};

	\node[ rectangle split, rectangle split parts=1, draw,   anchor=east] at (-1,2) (e) {\nodepart{one} $b$};
  	\node[ rectangle split, rectangle split parts=1, draw,   anchor=east] at (0,2) (d) {\nodepart{one} $a$};
  	\node[ rectangle split, rectangle split parts=1, draw,   anchor=east] at (1,2) (f) {\nodepart{one} $e$};
  	\node[ rectangle split, rectangle split parts=1, draw,   anchor=east] at (2,2) (g) {\nodepart{one} $d$};
	\node[rectangle split, rectangle split parts=1, draw,   anchor=east] at  (3.3,2) (h) {\nodepart{one} $f$};

	\node[ circle, draw, anchor=east] at (-2.2,0.3) (00) {$y$};
  	\node[ rectangle split, rectangle split parts=1, draw,   anchor=east] at (-1.5,0.3) (x) {\nodepart{one} $a$};
  	\node[ rectangle split, rectangle split parts=1, draw,   anchor=east] at (0,0.3) (a) {\nodepart{one} $f$};
  	\node[ rectangle split, rectangle split parts=1, draw,   anchor=east] at  (1.5,0.3) (b) {\nodepart{one} $e$};
	\draw (a) edge[out=-32,in=210,->] (b);
	\draw(a) edge[out=220, in=-40,->](x);
	\draw(x) edge[out=-32, in=210, ->](a);
	\draw(b) edge[out=220, in=-40,->](a);
	\draw (00) edge[out=0, in=180, ->](x);	

	\draw (x) edge[out=90, in=-90,->](d);
	\draw (a) edge[out=90, in=-90,->](h);
	\draw (b) edge[out=90, in=-90,->](f);

	\draw (k) edge[out=-130, in=90,->](e); 
	\draw (k) edge[out=-90, in=90, ->] (g);
	\draw (k) edge[out=90, in=-90, ->](l);

	\draw (kk) edge[out=260, in=180, ->] (h);
	\draw (kk) edge[out=-75, in=0, ->] (h);
	\draw (kk) edge[out=90, in=-90, ->](ll);

	\draw[dotted] (e) edge[out=110, in=-115, ->](k);
	\draw[dotted] (d) edge[out=80, in=-115, ->](k);
	\draw[dotted] (f) edge[out=110, in=-107, ->](k);
	\draw[dotted] (g) edge[out=110, in=-95, ->](k);

	\draw[dotted] (h) edge[out=110, in=-95, ->](kk);

	\draw (e) edge[out=-25, in=210, ->](d);
	\draw (d) edge[out=160, in=10,->](e);

	\draw (d) edge[out=-25, in=210, ->](f);
	\draw (f) edge[out=160, in=10,->](d);

	\draw (f) edge[out=-25, in=210, ->](g);
	\draw (g) edge[out=160, in=10,->](f);
	\node[anchor=east] at (3, 6) (capt) {The initial setup for pivot $y$};
\end{tikzpicture}
\end{minipage}
\begin{minipage}{.4\textwidth}
\centering
	\begin{tikzpicture}[thick,scale=0.6, every node/.style={scale=0.7}]
	\node[anchor=east] at (0.5,6) (kkk) {};
	\node[ rectangle split, rectangle split parts=1, draw,   anchor=east] at (0.5,4) (k) {\nodepart{one} $Q_1$};
	\node[anchor=east] at (-0.2,3.5) (first) {{\tiny $\bf{f}$}};
	\node[anchor=east] at (2.15,2.5) (last) {{\tiny $\bf{l}$}};
	\node[anchor=east] at (0.2,2.4) (current) {{\tiny $\bf{c}$}};

	\node[ rectangle split, rectangle split parts=1, draw,   anchor=east] at (3,4) (kk) {\nodepart{one} $Q_2$};
	\node[anchor=east] at (2.55,3.5) (firstb) {{\tiny $\bf{f}$}};
	\node[anchor=east] at (3.4,3.5) (lastb) {{\tiny $\bf{l}$}};
	\node[anchor=east] at (3.35,2.5) (currentb) {{\tiny $\bf{c}$}};

	\node[ rectangle split, rectangle split parts=1, draw,   anchor=east] at (-1,2) (e) {\nodepart{one} $a$};
  	\node[ rectangle split, rectangle split parts=1, draw,   anchor=east] at (0,2) (d) {\nodepart{one} $e$};
  	\node[ rectangle split, rectangle split parts=1, draw,   anchor=east] at (1,2) (f) {\nodepart{one} $b$};
  	\node[ rectangle split, rectangle split parts=1, draw,   anchor=east] at (2,2) (g) {\nodepart{one} $d$};
	\node[rectangle split, rectangle split parts=1, draw,   anchor=east] at  (3.3,2) (h) {\nodepart{one} $f$};

	\node[ circle, draw, anchor=east] at (-2.2,0.3) (00) {$y$};
  	\node[ rectangle split, rectangle split parts=1, draw,   anchor=east] at (-1.5,0.3) (x) {\nodepart{one} $a$};
  	\node[ rectangle split, rectangle split parts=1, draw,   anchor=east] at (0,0.3) (a) {\nodepart{one} $f$};
  	\node[ rectangle split, rectangle split parts=1, draw,   anchor=east] at  (1.5,0.3) (b) {\nodepart{one} $e$};
	\draw (a) edge[out=-32,in=210,->] (b);
	\draw(a) edge[out=220, in=-40,->](x);
	\draw(x) edge[out=-32, in=210, ->](a);
	\draw(b) edge[out=220, in=-40,->](a);
	\draw (00) edge[out=0, in=180, ->](x);	

	\draw (x) edge[out=90, in=-90,->](e);
	\draw (a) edge[out=90, in=-90,->](h);
	\draw (b) edge[out=90, in=-90,->](d);

	\draw (k) edge[out=-130, in=90,->](e); 
	\draw (k) edge[out=-70, in=90, ->] (g);
	\draw (k) edge[out=-90, in=90, ->] (d);

	\draw (kk) edge[out=240, in=180, ->] (h);
	\draw (kk) edge[out=-60, in=0, ->] (h);
	\draw (kk) edge[out=-90, in=90, ->](h);

	\draw[dotted] (e) edge[out=110, in=-115, ->](k);
	\draw[dotted] (d) edge[out=110, in=-115, ->](k);
	\draw[dotted] (f) edge[out=75, in=-80, ->](k);
	\draw[dotted] (g) edge[out=110, in=-95, ->](k);

	\draw[dotted] (h) edge[out=110, in=-95, ->](kk);

	\draw (e) edge[out=-25, in=210, ->](d);
	\draw (d) edge[out=160, in=10,->](e);

	\draw (d) edge[out=-25, in=210, ->](f);
	\draw (f) edge[out=160, in=10,->](d);

	\draw (f) edge[out=-25, in=210, ->](g);
	\draw (g) edge[out=160, in=10,->](f);
	\node[anchor=east] at (3, 6) (capt) {The new reordering of the subpartitions};
	\end{tikzpicture}
\end{minipage}
\begin{minipage}{.1\textwidth}
\centering
	\begin{tikzpicture}[thick,scale=0.6, every node/.style={scale=0.7}]
	\node[anchor=east] at (0.5,0) (kkk) {};
	\node[anchor=east] at (0.5,5) (l) {{\tiny null}};
	\node[ rectangle split, rectangle split parts=1, draw,   anchor=east] at (0.5,4) (k) {\nodepart{one} $Q_1$};
	\node[anchor=east] at (-0.2,3.5) (first) {{\tiny $\bf{f}$}};
	\node[anchor=east] at (0.5,3.5) (last) {{\tiny $\bf{l}$}};
	\node[anchor=east] at (0.5,4.5) (current) {{\tiny $\bf{c}$}};

	\node[anchor=east] at (1.7,5) (lll) {{\tiny null}};
	\node[ rectangle split, rectangle split parts=1, draw,   anchor=east] at (1.75,4) (kkk) {\nodepart{one} $Q_1'$};
	\node[anchor=east] at (1,3.5) (firsta) {{\tiny $\bf{f}$}};
	\node[anchor=east] at (2.15,2.5) (lasta) {{\tiny $\bf{l}$}};
	\node[anchor=east] at (1.75,4.5) (currenta) {{\tiny $\bf{c}$}};

	\node[ rectangle split, rectangle split parts=1, draw,   anchor=east] at (3,4) (kk) {\nodepart{one} $Q_2$};
	\node[anchor=east] at (2.55,3.5) (firstb) {{\tiny $\bf{f}$}};
	\node[anchor=east] at (3.3,3.5) (lastb) {{\tiny $\bf{l}$}};
	\node[anchor=east] at (3.35,2.5) (currentb) {{\tiny $\bf{c}$}};

	\node[ rectangle split, rectangle split parts=1, draw,   anchor=east] at (-1,2) (e) {\nodepart{one} $a$};
  	\node[ rectangle split, rectangle split parts=1, draw,   anchor=east] at (0,2) (d) {\nodepart{one} $e$};
  	\node[ rectangle split, rectangle split parts=1, draw,   anchor=east] at (1,2) (f) {\nodepart{one} $b$};
  	\node[ rectangle split, rectangle split parts=1, draw,   anchor=east] at (2,2) (g) {\nodepart{one} $d$};
	\node[rectangle split, rectangle split parts=1, draw,   anchor=east] at  (3.3,2) (h) {\nodepart{one} $f$};

	\draw (k) edge[out=-130, in=90,->](e); 
	\draw (k) edge[out=-90, in=90, ->] (d);

	\draw (kkk) edge[out=-130, in=90,->](f); 
	\draw (kkk) edge[out=-75, in=90, ->] (g);

	\draw (kk) edge[out=240, in=180, ->] (h);
	\draw (kk) edge[out=-60, in=0, ->] (h);
	\draw (kk) edge[out=-90, in=90, ->](h);

	\draw (k) edge[out=90, in=-90, ->](l);
	\draw (kkk) edge[out=90, in=-90, ->](lll);

	\draw[dotted] (e) edge[out=110, in=-115, ->](k);
	\draw[dotted] (d) edge[out=110, in=-115, ->](k);
	\draw[dotted] (g) edge[out=75, in=-80, ->](kkk);
	\draw[dotted] (f) edge[out=110, in=-95, ->](kkk);

	\draw[dotted] (h) edge[out=110, in=-95, ->](kk);

	\draw (e) edge[out=-25, in=210, ->](d);
	\draw (d) edge[out=160, in=10,->](e);

	\draw (f) edge[out=-25, in=210, ->](g);
	\draw (g) edge[out=160, in=10,->](f);
	\node[anchor=east] at (3, 6) (capt) {Splitting the subpartitions at $\bf{c}$};
\end{tikzpicture}
\end{minipage}
\caption{Processing pivot $y$}
\end{figure}
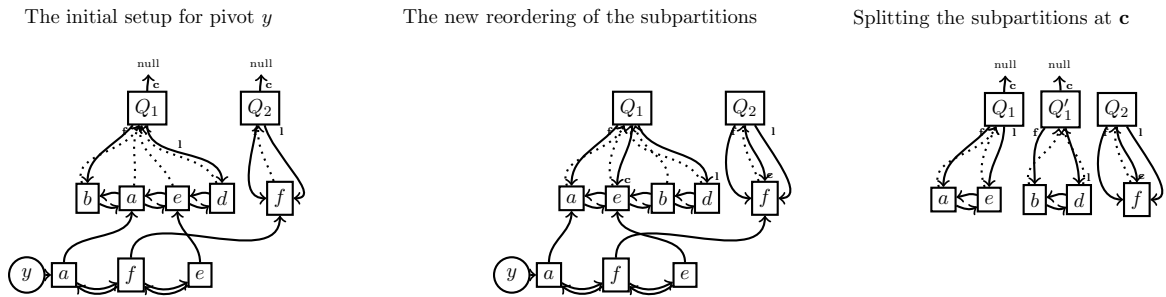
 Notice how $Q_2$ has all its pointers pointing at the same element $f$, therefore the current pointer $\bf{c}$ is not reset to NULL. The next pivot is $z$.
\begin{figure}[H]
\begin{minipage}{.3\textwidth}
	\begin{tikzpicture}[thick,scale=0.6, every node/.style={scale=0.7}]
	\node[anchor=east] at (0.5,6) (kkk) {};
	\node[anchor=east] at (0.5,5) (l) {{\tiny null}};
	\node[ rectangle split, rectangle split parts=1, draw,   anchor=east] at (0.5,4) (k) {\nodepart{one} $Q_1$};
	\node[anchor=east] at (-0.2,3.5) (first) {{\tiny $\bf{f}$}};
	\node[anchor=east] at (0.5,3.5) (last) {{\tiny $\bf{l}$}};
	\node[anchor=east] at (0.5,4.5) (current) {{\tiny $\bf{c}$}};

	\node[anchor=east] at (1.7,5) (lll) {{\tiny null}};
	\node[ rectangle split, rectangle split parts=1, draw,   anchor=east] at (1.75,4) (kkk) {\nodepart{one} $Q_1'$};
	\node[anchor=east] at (1,3.5) (firsta) {{\tiny $\bf{f}$}};
	\node[anchor=east] at (2.15,2.5) (lasta) {{\tiny $\bf{l}$}};
	\node[anchor=east] at (1.75,4.5) (currenta) {{\tiny $\bf{c}$}};

	\node[ rectangle split, rectangle split parts=1, draw,   anchor=east] at (3,4) (kk) {\nodepart{one} $Q_2$};
	\node[anchor=east] at (2.55,3.5) (firstb) {{\tiny $\bf{f}$}};
	\node[anchor=east] at (3.3,3.5) (lastb) {{\tiny $\bf{l}$}};
	\node[anchor=east] at (3.35,2.5) (currentb) {{\tiny $\bf{c}$}};

	\node[ rectangle split, rectangle split parts=1, draw,   anchor=east] at (-1,2) (e) {\nodepart{one} $a$};
  	\node[ rectangle split, rectangle split parts=1, draw,   anchor=east] at (0,2) (d) {\nodepart{one} $e$};
  	\node[ rectangle split, rectangle split parts=1, draw,   anchor=east] at (1,2) (f) {\nodepart{one} $b$};
  	\node[ rectangle split, rectangle split parts=1, draw,   anchor=east] at (2,2) (g) {\nodepart{one} $d$};
	\node[rectangle split, rectangle split parts=1, draw,   anchor=east] at  (3.3,2) (h) {\nodepart{one} $f$};

	\node[ circle, draw, anchor=east] at (-2.2,0.3) (00) {$z$};
  	\node[ rectangle split, rectangle split parts=1, draw,   anchor=east] at (-1.5,0.3) (x) {\nodepart{one} $b$};
  	\node[ rectangle split, rectangle split parts=1, draw,   anchor=east] at (0,0.3) (a) {\nodepart{one} $a$};
	\draw(a) edge[out=220, in=-40,->](x);
	\draw(x) edge[out=-32, in=210, ->](a);
	\draw (00) edge[out=0, in=180, ->](x);	

	\draw (x) edge[out=90, in=-90,->](f);
	\draw (a) edge[out=90, in=-90,->](e);
	
	\draw (k) edge[out=-130, in=90,->](e); 
	\draw (k) edge[out=-90, in=90, ->] (d);

	\draw (kkk) edge[out=-130, in=90,->](f); 
	\draw (kkk) edge[out=-75, in=90, ->] (g);

	\draw (kk) edge[out=240, in=180, ->] (h);
	\draw (kk) edge[out=-60, in=0, ->] (h);
	\draw (kk) edge[out=-90, in=90, ->](h);

	\draw (k) edge[out=90, in=-90, ->](l);
	\draw (kkk) edge[out=90, in=-90, ->](lll);

	\draw[dotted] (e) edge[out=110, in=-115, ->](k);
	\draw[dotted] (d) edge[out=110, in=-115, ->](k);
	\draw[dotted] (g) edge[out=75, in=-80, ->](kkk);
	\draw[dotted] (f) edge[out=110, in=-95, ->](kkk);

	\draw[dotted] (h) edge[out=110, in=-95, ->](kk);

	\draw (e) edge[out=-25, in=210, ->](d);
	\draw (d) edge[out=160, in=10,->](e);

	\draw (f) edge[out=-25, in=210, ->](g);
	\draw (g) edge[out=160, in=10,->](f);
	\node[anchor=east] at (3,6) (capt) {The initial setup for pivot $z$};
\end{tikzpicture}
\end{minipage}
\begin{minipage}{.3\textwidth}
\centering
	\begin{tikzpicture}[thick,scale=0.6, every node/.style={scale=0.7}]
	\node[anchor=east] at (0.5,6) (kkk) {};
	\node[ rectangle split, rectangle split parts=1, draw,   anchor=east] at (0.5,4) (k) {\nodepart{one} $Q_1$};
	\node[anchor=east] at (-0.2,3.5) (first) {{\tiny $\bf{f}$}};
	\node[anchor=east] at (0.65,3.5) (last) {{\tiny $\bf{l}$}};
	\node[anchor=east] at (-0.75,2.5) (current) {{\tiny $\bf{c}$}};

	\node[ rectangle split, rectangle split parts=1, draw,   anchor=east] at (1.75,4) (kkk) {\nodepart{one} $Q_1'$};
	\node[anchor=east] at (1.1,3.5) (firsta) {{\tiny $\bf{f}$}};
	\node[anchor=east] at (1.8,2.5) (lasta) {{\tiny $\bf{l}$}};
	\node[anchor=east] at (1.25,2.5) (currenta) {{\tiny $\bf{c}$}};

	\node[ rectangle split, rectangle split parts=1, draw,   anchor=east] at (3,4) (kk) {\nodepart{one} $Q_2$};
	\node[anchor=east] at (2.55,3.5) (firstb) {{\tiny $\bf{f}$}};
	\node[anchor=east] at (3.3,3.5) (lastb) {{\tiny $\bf{l}$}};
	\node[anchor=east] at (3.35,2.5) (currentb) {{\tiny $\bf{c}$}};

	\node[ rectangle split, rectangle split parts=1, draw,   anchor=east] at (-1,2) (e) {\nodepart{one} $a$};
  	\node[ rectangle split, rectangle split parts=1, draw,   anchor=east] at (0,2) (d) {\nodepart{one} $e$};
  	\node[ rectangle split, rectangle split parts=1, draw,   anchor=east] at (1,2) (f) {\nodepart{one} $b$};
  	\node[ rectangle split, rectangle split parts=1, draw,   anchor=east] at (2,2) (g) {\nodepart{one} $d$};
	\node[rectangle split, rectangle split parts=1, draw,   anchor=east] at  (3.3,2) (h) {\nodepart{one} $f$};

	\node[ circle, draw, anchor=east] at (-2.2,0.3) (00) {$z$};
  	\node[ rectangle split, rectangle split parts=1, draw,   anchor=east] at (-1.5,0.3) (x) {\nodepart{one} $b$};
  	\node[ rectangle split, rectangle split parts=1, draw,   anchor=east] at (0,0.3) (a) {\nodepart{one} $a$};
	\draw(a) edge[out=220, in=-40,->](x);
	\draw(x) edge[out=-32, in=210, ->](a);
	\draw (00) edge[out=0, in=180, ->](x);	

	\draw (x) edge[out=90, in=-90,->](f);
	\draw (a) edge[out=90, in=-90,->](e);
	
	\draw (k) edge[out=-130, in=90,->](e); 
	\draw (k) edge[out=-60, in=90, ->] (d);

	\draw (kkk) edge[out=-130, in=90,->](f); 
	\draw (kkk) edge[out=-75, in=90, ->] (g);

	\draw (kk) edge[out=240, in=180, ->] (h);
	\draw (kk) edge[out=-60, in=0, ->] (h);
	\draw (kk) edge[out=-90, in=90, ->](h);

	\draw (k) edge[out=-90, in=75, ->](e);
	\draw (kkk) edge[out=-90, in=75, ->](f);

	\draw[dotted] (e) edge[out=110, in=-115, ->](k);
	\draw[dotted] (d) edge[out=110, in=-75, ->](k);
	\draw[dotted] (g) edge[out=75, in=-80, ->](kkk);
	\draw[dotted] (f) edge[out=110, in=215, ->](kkk);

	\draw[dotted] (h) edge[out=110, in=-95, ->](kk);

	\draw (e) edge[out=-25, in=210, ->](d);
	\draw (d) edge[out=160, in=10,->](e);

	\draw (f) edge[out=-25, in=210, ->](g);
	\draw (g) edge[out=160, in=10,->](f);
	\node[anchor=east] at (2, 6) (capt) {The new reordering of};
	\node[anchor=east] at (2, 5.5) (capti) {the subpartitions};
\end{tikzpicture}
\end{minipage}
\begin{minipage}{.1\textwidth}
\centering
	\begin{tikzpicture}[thick,scale=0.6, every node/.style={scale=0.7}]
	\node[anchor=east] at (0.5,0) (kkk) {};
	\node[ rectangle split, rectangle split parts=1, draw,   anchor=east] at (0.5,4) (p) {\nodepart{one} $Q_1$};
	\node[anchor=east] at (0,3.5) (fp) {{\tiny $\bf{f}$}};
	\node[anchor=east] at (0.8,3.5) (lp) {{\tiny $\bf{l}$}};
	\node[anchor=east] at (0.85,2.5) (cp) {{\tiny $\bf{c}$}};

	\node[anchor=east] at (1.75,5) (l) {{\tiny null}};
	\node[ rectangle split, rectangle split parts=1, draw,   anchor=east] at (1.75,4) (o) {\nodepart{one} $Q_2$};
	\node[anchor=east] at (1.25,3.5) (op) {{\tiny $\bf{f}$}};
	\node[anchor=east] at (2.05,3.5) (ol) {{\tiny $\bf{l}$}};
	\node[anchor=east] at (1.8,4.5) (oc) {{\tiny $\bf{c}$}};

	\node[ rectangle split, rectangle split parts=1, draw,   anchor=east] at (3,4) (n) {\nodepart{one} $Q_3$};
	\node[anchor=east] at (2.5,3.5) (np) {{\tiny $\bf{f}$}};
	\node[anchor=east] at (3.3,3.5) (nl) {{\tiny $\bf{l}$}};
	\node[anchor=east] at (3.35,2.5) (nc) {{\tiny $\bf{c}$}};

	\node[anchor=east] at (4.25,5) (ll) {{\tiny null}};
	\node[ rectangle split, rectangle split parts=1, draw,   anchor=east] at (4.25,4) (k) {\nodepart{one} $Q_4$};
	\node[anchor=east] at (3.75,3.5) (fk) {{\tiny $\bf{f}$}};
	\node[anchor=east] at (4.55,3.5) (lk) {{\tiny $\bf{l}$}};
	\node[anchor=east] at (4.3,4.5) (ck) {{\tiny $\bf{c}$}};

	\node[ rectangle split, rectangle split parts=1, draw,   anchor=east] at (6,4) (m) {\nodepart{one} $Q_5$};
	\node[anchor=east] at (5.5,3.5) (fm) {{\tiny $\bf{f}$}};
	\node[anchor=east] at (6.3,3.5) (lm) {{\tiny $\bf{l}$}};
	\node[anchor=east] at (6.35,2.5) (cm) {{\tiny $\bf{c}$}};

	\draw (o) edge[out=90, in=-90, ->](l);
	\draw (k) edge[out=90, in=-90, ->](ll);

	\node[ rectangle split, rectangle split parts=1, draw,   anchor=east] at (0.8,2) (a) {\nodepart{one} $a$};
	\draw (p) edge[out=240, in=135, ->](a); 
	\draw (p) edge[out=-60, in=45, ->](a);
	\draw (p) edge[out=-90, in=90, ->](a); 

  	\node[ rectangle split, rectangle split parts=1, draw,   anchor=east] at (2.05,2) (e) {\nodepart{one} $e$};
	\draw (o) edge[out=240, in=135, ->](e); 
	\draw (o) edge[out=-60, in=45, ->](e);

  	\node[ rectangle split, rectangle split parts=1, draw,   anchor=east] at (3.3,2) (b) {\nodepart{one} $b$};
	\draw (n) edge[out=240, in=135, ->](b); 
	\draw (n) edge[out=-60, in=45, ->](b);
	\draw (n) edge[out=-90, in=90, ->](b); 

  	\node[ rectangle split, rectangle split parts=1, draw,   anchor=east] at (4.55,2) (d) {\nodepart{one} $d$};
	\draw (k) edge[out=240, in=135, ->](d); 
	\draw (k) edge[out=-60, in=45, ->](d);
	
	\node[rectangle split, rectangle split parts=1, draw,   anchor=east] at  (6.3,2) (f) {\nodepart{one} $f$};
	\draw (m) edge[out=240, in=135, ->](f); 
	\draw (m) edge[out=-60, in=45, ->](f);
	\draw (m) edge[out=-90, in=90, ->](f); 

	\draw[dotted] (a) edge[out=110, in=-100, ->](p);
	\draw[dotted] (e) edge[out=110, in=-100, ->](o);
	\draw[dotted] (b) edge[out=110, in=-100, ->](n);
	\draw[dotted] (d) edge[out=110, in=-100, ->](k);
	\draw[dotted] (f) edge[out=110, in=-100, ->](m);
\node[anchor=east] at (4,6) (capt) {splitting the subpartitions at $\bf{c}$};
\node[anchor=east] at (6, 5.5) (capti) {(renamed $Q_1, Q_2$, ..., $Q_5$ for simplicity)};
\end{tikzpicture}
\end{minipage}
\caption{Processing pivot $z$}
\end{figure}
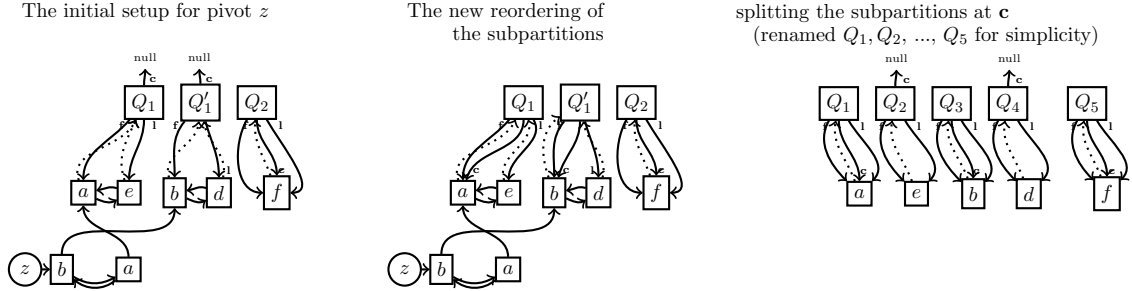
The refinement of $P = \{b, a, f, e, d\}$ is therefore $\tau = \{a, e, b, d, f\}$.
\\

Given the data structure we used, this reordering (i.e., refinement) of $P$ takes linear time. Splitting a (sub)partition $Q_i$ takes at most $\mathcal{O}(|N(v) \cap Q_i|)$ steps for every $Q_i$ with $\bf{c} \neq $NULL. For each $Q_i$, vertices between $\bf{f}$ and $\bf{c}$ are moved into a new subpartition $Q_j$, and $Q_i$ becomes $Q_i \backslash Q_j$, and thus, creating $Q_j$ and $Q_i = Q_i \backslash Q_j$ takes at most $\mathcal{O}(|N(v) \cap Q_i|)$ steps. Since this operation is performed on every $w \in N(v) \cap Q_i$, for every $v \in S$, we perform at most $\mathcal{O}(d_v)$ steps of refinement, thus a total of the number of edges going from vertices in $S$ to $P$.
\subsubsection{Algorithm 5:} By Lemma 5, step 2 is computed in linear time. As was just shown, $\emph{Refine}$ is linear in the size of the edges between $S_i$ and $P_i$, therefore summing over all $p$ partition classes, i.e., steps 4 to 8 in $\emph{CCLexDFS}$, we get $\Sigma|S_i| \in \mathcal{O}(m)$. The refinement is thus computed in $\mathcal{O}(m + n)$ time. Similarly, we just showed that $UpdatePivots$ is linear in the number of edges between $P_j$ and any class $P_{i > j}$; consequently looping over the $p$ partitions, $UpdatePivots$ takes a total of $\mathcal{O}(m + n)$ time. All the remaining steps in Algorithm 5 are clearly linear. This, together with the results of Section 4, complete the proof of Theorem 2.
\subsubsection{LexDFS$^+$:} Recall that to obtain a $LexDFS^{+}$, we just need to reorder the partition classes to maintain $\sigma^-$'s order. Since every vertex knows the partition class it belongs to, it suffices to sweep through $\sigma$ from right to left and recreate the partition classes (prior to the refinement), which takes $\mathcal{O}(n)$ time. Once this is done, the rest of the analysis remains the same on the newly ordered partition classes, therefore:
\begin{proposition}
Let $G(V, E)$ be a cocomparability graph and $\sigma$ a cocomparability ordering of $G$, the ordering $\tau = $ LexDFS$^+(\sigma)$ can be computed in $\mathcal{O}(m+n)$ time. 
\end{proposition}
\end{document}